\documentclass[11pt,fleqn]{amsart}

\usepackage{amsmath,amsthm}
\usepackage{amssymb}
\usepackage{dsfont}
\usepackage{amsfonts}
\usepackage{enumerate}
\usepackage{mathrsfs}
\usepackage{graphicx}
\usepackage{color}
\usepackage[round]{natbib}

\usepackage{abbreviationsSSC-ext}

\begin{document}

\title[Spatial sign correlation]{Spatial sign correlation}
\author[A. D{\"u}rre]{Alexander D{\"u}rre}
\author[D. Vogel]{Daniel Vogel}
\author[R. Fried]{Roland Fried}
\today

\address{Fakult\"at Statistik, Technische Universit\"at Dortmund, 
44221 Dortmund, Germany}
\email{alexander.duerre@tu-dortmund.de}
\email{daniel.vogel@tu-dortmund.de}
\email{fried@statistik.tu-dortmund.de}



\begin{abstract}
\small
A new robust correlation estimator based on the spatial sign covariance matrix (SSCM) is proposed. We derive its asymptotic distribution and influence function at elliptical distributions. Finite sample and robustness properties are studied and compared to other robust correlation estimators by means of numerical simulations.
\end{abstract}

\maketitle


\renewcommand{\thefootnote}{\fnsymbol{footnote}}
\section{Introduction}
\label{sec:intro}

The problem studied in this article is robust and high-dimensional correlation estimation. 
By \emph{robust} we mean insusceptible to outliers and erroneous observations, that is, we examine alternatives to the commonly used, but highly non-robust Pearson correlation. Over the last few decades, many robust multivariate scatter estimators --- and consequently robust correlation matrix estimators --- have been proposed, see \cite{Maronna2006} for a review. Much attention has been paid to affine equivariant estimators.
If we denote by $\X_n = (\bX_1,...,\bX_n)^T$  the $n \times p$ data matrix containing the $p$-dimensional observations $\bX_1,...,\bX_n$  as rows, then the data set $\Y_n = \X_n A^T + \bEins_n \bb^T$ is obtained by applying the affine linear transformation $\x \mapsto A \x + \bb$ to each data point. An affine equivariant scatter estimator $\hS_n$ satisfies $\hS_n(\Y_n) =  A \hS_n(\X_n) A^T$ for any $\bb \in \R^p$ and any full rank square matrix $A$, i.e.\ it behaves as the covariance matrix under linear transformations of the data.

The second attribute \emph{high-dimensional} means two things: being fast to compute, also in high-dimensions, and being defined also for sparse, high-dimensional data, i.e.\ in the $p > n$ situation. 
Both properties basically prohibit robust, affine equivariant estimators: they are usually hard to compute in high dimensions, and they are not defined in the $p > n$ setting --- or coincide with a multiple of the sample covariance matrix \citep{Tyler2010} and are thus not robust. In fact, both requirements suggest the use of pairwise correlation estimators. In a pairwise correlation estimate $\hat{P}_n \in \Rpp$ each entry $\hat{\rho}_{i,j}$ is computed only from the $i$th and the $j$th coordinate of the data, implying that the computing time increases quadratically with $p$.

The price one usually has to pay for dropping affine equivariance and resorting to pairwise correlation estimators is the loss of non-negative definiteness of the matrix estimate $\hat{P}_n$.
For example, many nonparametric correlation matrix estimators (see Section \ref{sec:anal.comp}) are based on an initial scatter matrix estimate which is non-negative definite, but not affine equivariant. The loss of non-negative definiteness occurs when a component-wise transformation is applied to render the entries consistent for the moment correlation.  
However, in applications where non-negative definiteness is important, one can ``orthogonalize'' the matrix estimate as suggested by \citet{Maronnazamar2002}, which involves an eigenvalue decomposition.

The new proposal is based on the spatial sign covariance matrix (SSCM). This is the covariance matrix of the projections of the centered observations onto the $p$-dimensional unit sphere. This scatter estimator is of frequent use in multivariate data analysis due to its robustness. Since every observation is basically trimmed to length 1, the impact of any contamination is bounded. It is known that within symmetric data models, the SSCM consistently estimates the eigenvectors of the covariance matrix, but not the eigenvalues. In fact, the connection between the eigenvalues of the population SSCM and the covariance matrix is an open problem. We solve this problem for the special case of two-dimensional elliptical distributions. This enables us to robustly estimate a two-dimensional covariance matrix (up to scale) based on the SSCM and hence devise a correlation estimator, which we call \emph{spatial sign correlation}. We further derive the asymptotic distribution of the SSCM and the spatial sign correlation and compute the influence function of the latter.

The main advantage of the new estimator is its simplicity. It is very fast to compute, it requires neither an iterative algorithm nor any ranking or sorting of the data. It is furthermore distribution-free within the elliptical model, it behaves equally well for very heavy-tailed and very peaked distributions, which is true for hardly any other robust scatter estimator.\footnote{For these statements to be true, the SSCM has to based on an appropriate location estimator.}

The paper has two parts: In part 1, consisting of Sections \ref{sec:sscm} -- \ref{sec:correlation}, we develop the spatial sign correlation estimator and derive its asymptotics. Being aware that this estimator is one out of many that were introduced for similar purposes, the second part, consisting of Sections \ref{sec:anal.comp} and \ref{sec:num.comp}, gathers together analytic results about a variety of alternatives and compares them in an elaborate simulation study to provide some guidance within the ever increasing number of robust correlation estimates. All proofs are deferred to the appendix. We close this section by introducing some recurrent terms and notation. 

In order to study the properties of the new estimator analytically we will assume the data to stem from the elliptical model.
A continuous distribution $F$ on $\mathds{R}^p$ is said to be \emph{elliptical} if it has a Lebesgue-density $f$ of the form
\begin{equation} \label{density}
	f(\bx) = \det(V)^{-\frac{1}{2}} g\big((\bx-\bmu)^T V^{-1} (\bx-\bmu)\big)
\end{equation}
for some $\bmu \in \mathds{R}^p$ and symmetric, positive definite $p \times p$ matrix $V$. We call $\bmu$ the \emph{location} or \emph{symmetry center} and $V$ the \emph{shape matrix}, since it describes the shape of the elliptical contour lines of the density.  The class of all continuous elliptical distributions $F$ on $\mathds{R}^p$ having these parameters is denoted by $\mathscr{E}_p(\bmu, V)$.
The shape matrix $V$ is unique only up to scale, that is, $\Ee_p(\bmu, V) = \Ee_p(\bmu, c V)$ for any $c > 0$. For scale-free functions of $V$, such as correlations, which we consider here, this ambiguity is irrelevant. A common view on the \emph{shape} of an elliptical distribution is to treat it as an equivalence class of positive definite random matrices being proportional to each other. We adopt this notion here: in the results of this exposition, $V$ can be any representative from its equivalence class. For example, if second moments exist, one can always take the covariance matrix  --- or any suitably scaled multiple of it. However, the results are more general, the existence of second --- or even first --- moments is not required. Throughout the paper we let 
\be \label{evd} 
	 V = U \Lambda U^T
\ee 
denote an eigenvalue decomposition of $V$, where $U$ is an orthogonal matrix containing the eigenvectors of $V$ as columns and $\Lambda \diag(\lambda_1,...,\lambda_p)$ is such that $0 < \lambda_p \le \ldots \le \lambda_1$. We use $||\cdot||$ to denote the $L_2$ norm of a vector.


\section{The spatial sign covariance matrix}
\label{sec:sscm}

We define the spatial sign covariance matrix of a multivariate distribution and derive its connection to the shape matrix $V$ in case of a two-dimensional elliptical distribution. For $\bx \in \R^p$ define the \emph{spatial sign} $\bs(\bx)$ of $\bx$ as
	$\bs(\bx) = \bx/||\bx||$ if $\bx \neq \bNull$ and $\bs(\bx) = \bNull$ otherwise. Let $\bX$ be a $p$-dimensional random vector ($p \ge 2$) having distribution $F$. We call  
\[
	\bmu(F) = \bmu(\bX) = \argmin_{\mbox{\scriptsize $\bmu$} \in \R^p} \E \left( ||\bX - \bmu|| - ||\bX|| \right)
\]
the \emph{spatial median} and, following the terminology of \citet{Visuri2000}, 
\[
	S(F) = S(\bX) = \E\left( \bs(\bX-\bmu) \bs(\bX-\bmu)^T \right)
\]
the \emph{spatial sign covariance matrix (SSCM)} of $F$ (or $\bX$). If there is no unique minimizing point of $\E \left( ||\bX - \bmu|| - ||\bX|| \right)$, then $\bmu(F)$ is the barycenter of the minimizing set. This may only happen if $F$ is concentrated on a line. For results on existence and uniqueness of the spatial median see \citet{Haldane1948}, \citet{Kemperman1987}, \citet{Milasevic1987} or \citet{Koltchinskii2000}. 
If the first moments of $F$ are finite, then the spatial median allows the more descriptive characterization as 
$\argmin_{\mbox{\scriptsize $\bmu$} \in \R^p} \E ||\bX - \bmu||$. The spatial median always exists.


Let $\X_n = (\bX_1,\dots,\bX_n)^T$ be a data sample of size $n$, where the $\bX_i$, $i = 1,...,n$, are i.i.d., each with distribution $F$. Define
\[
	\hat{S}_n (\X_n ;  \dt) = \ave_{i = 1,...,n} \bs(\bX_i-\dt)\bs(\bX_i-\dt)^T
\]
where $\dt \in \R^p$. 
Choosing $\dt = \bmu(F)$, we call the estimator $\hat{S}_n (\X_n ;  \bmu(F))$ the \emph{empirical SSCM with known location}. 
However, the location is usually unknown, and $\dt$ has to be replaced by a suitable location estimator $(\bT_n)_{n \in \N}$, and we refer to $\hat{S}_n (\X_n ;  \bT_n)$ as the \emph{empirical SSCM with unknown location}. The canonical location functional in this case is the \emph{(empirical) spatial median}
\[
	\hat{\bmu}_n = \hat{\bmu}_n(\X_n) = \min_{\mbox{\scriptsize $\bmu$} \in \R^p} \sum_{i=1}^n || \bX_i - \bmu ||.
\]
Under regularity conditions (the data points do not lie on a line and none of them coincides with $\hat{\bmu}_n$, see \citet{Kemperman1987}, p.~228) the spatial signs w.r.t.\ the empirical spatial median are centered, i.e.\ $\ave_{i=1}^n \bs(\bX_i - \hat{\bmu}_n)  = \bNull$. 
Hence, the empirical spatial sign covariance matrix $\hat{S}_n (\X_n ; \hat{\bmu}_n)$ is indeed the covariance matrix of the spatial signs --- if the latter are taken w.r.t.\ the spatial median. Our first proposition draws a connection between the shape of an elliptical distribution and the corresponding spatial sign covariance matrix.
\begin{proposition} \label{prop:S}
Let $F \in \Ee_p(\bmu,V)$ and $V = U \Lambda U^T$ denote an eigenvalue decomposition of $V$ with $0 < \lambda_p \le \ldots \le \lambda_1$. Then
\begin{enumerate}[(1)]
\item  \label{S 1} $\bmu(F) = \bmu$ and
\item \label{S 2} $S(F) = U \Delta U^T$, where $\Delta = \diag(\delta_1,...,\delta_p)$ is a diagonal matrix with $0 < \delta_p \le \ldots \le \delta_1$. 
\item \label{S 3} If $p = 2$, then $\delta_j = \sqrt{\lambda_j}/(\sqrt{\lambda_1} + \sqrt{\lambda_2})$, $j = 1, 2$.
\end{enumerate}
\end{proposition}
The proof is given in the appendix. Part (\ref{S 2}) of Proposition \ref{prop:S} states that the SSCM $S(F)$ and the shape matrix $V$ have the same eigenvectors and the same order of the corresponding eigenvalues. This has been known for some time, and the use of the SSCM has been proposed to robustify such multivariate analyses that are based on this information only, most notably principal component analysis, \citep{Marden1999, Locantore1999, Croux2002, Gervini2008}. Other such applications are direction-of-arrival estimation \citep{Visuri2001a}, or testing sphericity in the elliptical model \citep{Sirkia2009}. 
Part (\ref{S 3}) enables us to reconstruct the whole shape matrix $V$ from $S(F)$ in dimension $p =2$. Thus the SSCM can be directly employed for applications that rely on the shape information, but do not require any knowledge about the overall scale, most notably correlations. This result seems to be quite recent. It appears in a similar form in \citet{Croux2010} and has also been used by \citet{Vogel2008}, but neither of these articles provide a proof. 

The next result concerns the asymptotic behavior of the empirical SSCM. It is formulated using the $\vec$ operator, which stacks the columns of a matrix from left to right underneath each other, and the Kronecker product $\otimes$ \citep[e.g.][Sec.~2]{Magnus1999}. Both are connected by the identity $\vec(ABC) = (C^T \otimes A) \vec B$. 
\begin{proposition} \label{prop:S.hat}
Let $\bX, \bX_1,\ldots,\bX_n$ be i.i.d.\ random vectors with distribution $F$ satisfying $\E||\bX-\bmu||^{-1}<\infty$ and $(\bT_n)_{n\in \mathbb{N}}$ a sequence of random variables converging almost surely to $\bmu(F)$. Then, as $n \rightarrow \infty$, we have
\begin{enumerate}[(1)]
\item \label{S.hat 1}
	$\hS_n(\X_n; \bT_n) \asc  S(F)$, and
\item \label{S.hat 2}
if furthermore 
	$\sqrt{n}||\bT_n-\bmu||$ converges in distribution, 
	$\mathbb{E}\big\{ ||\bX-\mu||^{-3/2} \big\} <\infty$ and 
	$(\bX-\bmu)\eid -(\bX-\bmu)$, then $\hS_n(\X_n; \bT_n)$ is asymptotically normal, i.e.\ there is a non-negative definite $p^2 \times p^2$ matrix $W_{\hS}$ such that
	\[
		\sqrt{n} \vec \big\{\hS_n(\X_n; \bT_n) - S(F)\big\} \cid N_{p^2}\left(\bNull, W_{S}\right).
	\]
\item \label{S.hat 3}
	If additionally $F \in \Ee_2(\bmu, V)$, then 
	\[
		\displaystyle W_{S} = \frac{-\lambda_1\lambda_2+\frac{1}{2}\sqrt{\lambda_1\lambda_2}(\lambda_1+\lambda_2)}{(\lambda_1-\lambda_2)^2} (U \otimes U) W_0 (U \otimes U)^T
		\] with		
			\[
			W_0 =
			\begin{pmatrix}
			1 & 0 & 0 & -1 \\
			0 & 1 & 1 &  0 \\
			0 & 1 & 1 &  0 \\
			-1 & 0 & 0 & 1 \\
		\end{pmatrix}.
	\]
\end{enumerate}
\end{proposition}
Parts (\ref{S.hat 1}), (\ref{S.hat 2}) are proved in \citet{Duerre2014}, where also alternative assumptions to guarantee strong consistency and asymptotic normality are given, The mentioned regularity conditions are rather mild, they are fulfilled for elliptical distributions with bounded density. 


\section{A spatial sign based correlation estimator}
\label{sec:correlation}

In the following let $\bX_i=(X_i,Y_i)^T$, $i=1,\ldots,n$, be an i.i.d.\ sample from $F \in \Ee_2(\bmu,V)$. Denoting the entries of $V$ by $v_{ij}$, we want to estimate the parameter
\[
	\rho = v_{12}/\sqrt{v_{11}v_{22}}.
\]
We call $\rho$ the \emph{generalized correlation coefficient} of the elliptical distribution $F$, since it coincides with the correlation coefficient if second moments are finite.
In a slight abuse of notation, we will refer to $\rho$ simply as the correlation (coefficient) of $F$ in the following. Propositions \ref{prop:S} and \ref{prop:S.hat} from the previous section give rise to an estimator of $\rho$ constructed as follows: compute the SSCM $\hat{S}_n = \hat{S}_n (\X_n ; \hat{\bmu}_n)$, perform an eigenvalue decomposition $\hS_n = \hat{U}_n \hat\Delta_n \hat{U}_n^T$ with 
$\hat\Delta_n = \diag(\hat\delta_1,\hat\delta_2)$ and compute the matrix $\hat{V}_n = \hat{U}_n \hat\Lambda_n \hat{U}_n^T$ with $\hat\Lambda_n = \diag(\hat\lambda_1,\hat\lambda_2)$ and $\hat\lambda_1 = \hat\delta_1/\hat\delta_2$, \, $\hat\lambda_2 = \hat\delta_2/\hat\delta_1$.\footnote{The overall scaling of $\hat{V}_n$ is, of course, irrelevant for the correlation, and its eigenvalues 
 $\hat\lambda_1$ and $\hat\lambda_2$ may as well be chosen differently. Their ratio has to satisfy $\hat\lambda_1/\hat\lambda_2 = (\hat\delta_1/\hat\delta_2)^2$.}
Finally compute the correlation coefficient from the matrix $\hat{V}_n$, i.e.\ let 
$\hat\rho_n = \hat{v}_{12}/\sqrt{\hat{v}_{11}\hat{v}_{22}}$. 
In dimension two, the eigenvalue decomposition can be computed explicitly with justifiable effort, and we obtain the following explicit expression for the thus defined estimator:
\[
	\hat{\rho}_n =\frac{c\hs_{12} b}{\sqrt{(\hs_{12}^2+b^2)^2+(\hs_{12}cb)^2}},
\]
where
\be \label{eq:cdb}
	c=\frac{2d-1}{d(1-d)}, \quad  d=\frac{1}{2}+\sqrt{(\hs_{11}-\frac{1}{2})^2+\hs_{12}^2},\quad  b=d-\hs_{11}
\ee
and $\hat{s}_{ij}$ denote the entries of $\hat{S}_n$. We call $\hat\rho_n$ the \emph{spatial sign correlation coefficient}. This must not be confused with the correlation of the spatial signs of the observations. This would be $\hat\rho_{\rm SSCM} = \hat{s}_{12}/\sqrt{\hat{s}_{11}\hat{s}_{22}}$. Also note that  knowing $\hat\rho_{\rm SSCM}$ alone is not sufficient for computing $\hat\rho_n$. Despite the rather lengthy definition of $\hat\rho_n$, its asymptotic variance has a surprisingly simple form. 
\begin{proposition} \label{prop:Skor}Let $F \in \Ee_2(\bmu, V)$ have a bounded density at $\bmu$. Then, as $n \to \infty$, 
\begin{enumerate}[(1)]
\item \label{Skor1}
	$\hat{\rho}_n \asc  \rho$, and
\item \label{prop:Skor:2}
	$\displaystyle \sqrt{n} (\hat{\rho}_n - \rho)\cid N\left(0, \, (1-\rho^2)^2 + \frac{1}{2}\left(a+a^{-1}\right)(1-\rho^2)^{3/2}\right),$ \
	where  \\ $a=\sqrt{v_{11}/v_{22}}$ is the root of the ratio of the diagonal elements of $V$.
\end{enumerate}
\end{proposition}
Proposition \ref{prop:Skor} (\ref{prop:Skor:2}) gives the asymptotic variance $ASV(\hat\rho_n)$ as a function of the true correlation $\rho$ and the ratio of the diagonal elements of the shape matrix $V$. The elliptical generator $g$, cf.~(\ref{density}), does not enter, which may be phrased as ``$\hat\rho_n$ is asymptotically distribution-free within the elliptical model''. It is furthermore consistent and asymptotically normal without any moment condition.

For fixed $\rho$, the asymptotic variance $ASV(\hat\rho_n)$ is minimal for equal marginal variances, but can get arbitrarily large for heteroscedastic data. It is therefore advisable to apply this estimator to standardized data, i.e.\ the components should be divided beforehand by a scale measure to yield equally dispersed margins. Margin-wise standardization generally should be administered with caution in multivariate data analysis, since it  changes the shape, e.g., the direction of the eigenvectors, and will alter the results of, e.g., a principal component analysis. The inefficiency of the spatial sign covariance matrix at strongly ``shaped'', i.e.\ non-spherical, distributions has led to criticism regarding its use for robust principal component analysis, where a strong ``shapedness'' is the working assumption, cf.\ e.g.\ Remark~5.1 in \citet{Bali2011}. 
We define shapedness as deviation from sphericity (and measure it for instance by the condition number of $V$). There are two sources that contribute to the shapedness: collinearity and heteroscedasticity. The formula in Proposition \ref{prop:Skor} (\ref{prop:Skor:2}) nicely visualizes the individual influences of these two sources of shapedness on the asymptotic variance of $\hat\rho_n$. Since we are interested in correlation --- a function of the shape that is invariant with respect to margin-wise scale changes ---, we can avoid the inefficiency due to the heteroscedasticity by margin-wise standardisation. 

Technical conditions that ensure the asymptotic equivalence of such a two-step procedure 
to the spatial sign correlation estimation at spherical distributions are yet to be established, but by heuristic arguments we can work for all practical purposes with an asymptotic variance of  
\[
	ASV(\hat\rho_n)	 = (1-\rho^2)^2 + (1-\rho^2)^{3/2}.
\]
In light of robustness, we recommend to use a highly robust scale estimator for standardization, such as the MAD or the $Q_n$ (see also next Section). Both have a breakdown point of 1/2, a property which they share with the spatial sign covariance matrix \citep{Croux2010}. The thus obtained two-stage correlation estimator is highly robust, but we refrain from considering breakdown points of correlation estimators, see the discussion and the rejoinder of \citet{Davies2005}.

In the next section we will compare several correlation estimators with respect to their efficiency at the normal model. As a first glimpse in this direction, we recall the asymptotic variance of the Pearson correlation $\hat\rho_{\mathrm{Pea}}$ at elliptical distributions
\[
 	ASV(\hat\rho_{\mathrm{Pea}}) = \left( 1 + \frac{\kappa}{3}\right) \left( 1 - \rho^2\right)^2,
\]
where $\kappa$ is the excess kurtosis of the components of $F$. The asymptotic relative efficiency of $\hat\rho_n$ with respect to 
$\hat\rho_{\rm Pea}$,
\[
	ARE(\hat\rho_n,\hat\rho_{\rm Pea}) = \frac{	ASV(\hat\rho_{\mathrm{Pea}})}{ASV(\hat\rho_n)} = \frac{1 + \kappa/3}{1 + \frac{1}{2}(a+a^{-1})(1-\rho^2)^{-1/2}},
\]
is depicted in Figure~\ref{fig:1}.
\begin{figure}
\begin{center}
\includegraphics[width=0.6\textwidth]{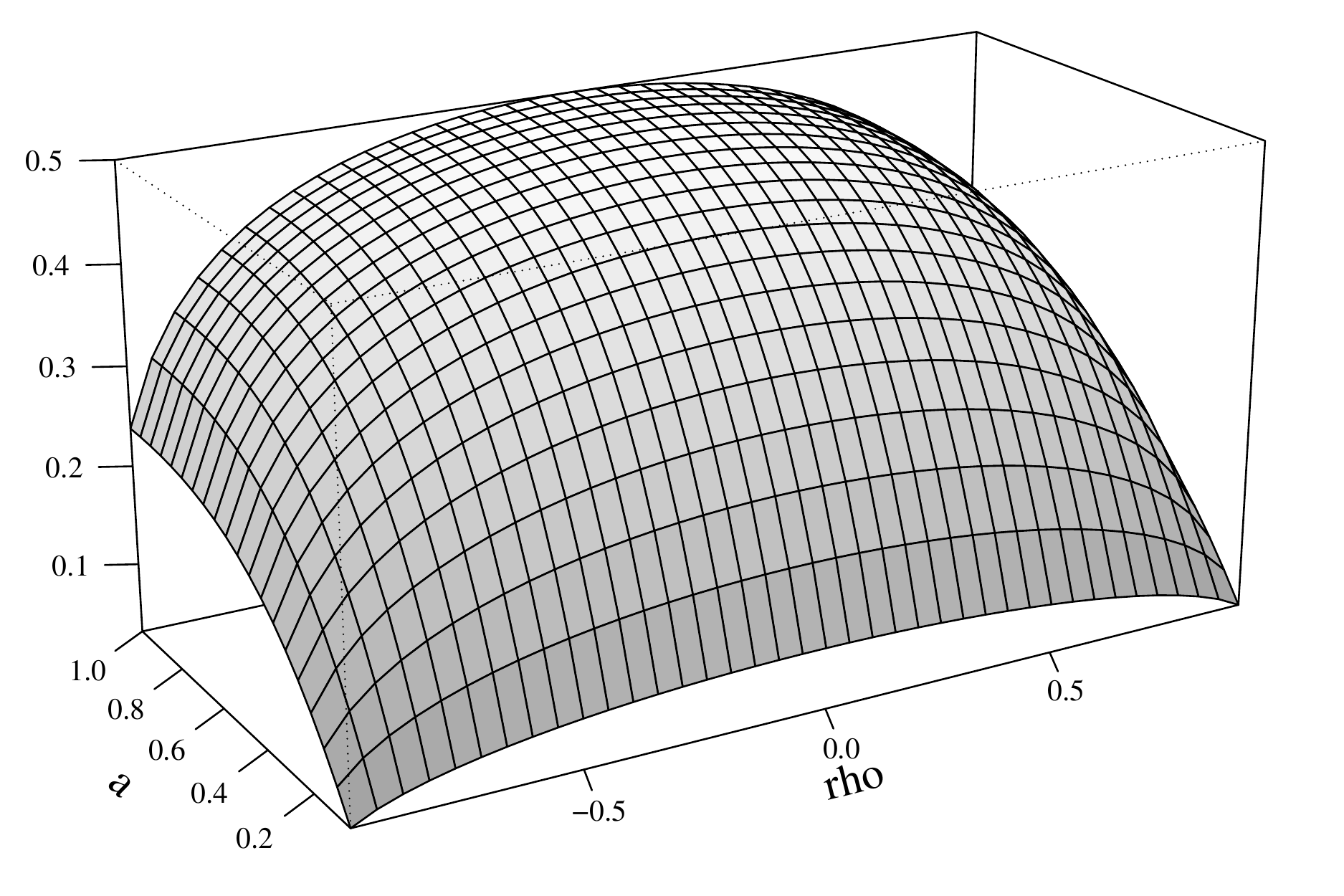}
\caption[efficiency]{The asymptotic relative efficiency of $\hat{\rho}$ with respect to the empirical correlation under normality as a function of $\rho$ and $a = \sqrt{v_{11}/v_{22}}$.}
\label{fig:1}
\end{center}
\end{figure}

At normality, the maximum 1/2 is attained for $a=1$ and $\rho =0$. If we fix $a =1$,
the asymptotic relative efficiency declines with increasing $|\rho|$, even tending to 0 for $|\rho|\rightarrow 1$. But it declines very slowly, for $|\rho| < 0.7 $ it stays above 0.4. 
Under heavy-tailed distributions, however, the spatial sign correlation can be more efficient than the Pearson correlation. Specifically, 	$ARE(\hat\rho_n,\hat\rho_{\rm Pea}) \geq 1$ if $\kappa \geq (3/2)(a+a^{-1})/\sqrt{1-\rho^2}$. For instance, with the kurtosis of the $t_\nu$ distribution being $6/(\nu - 4)$, the spatial sign correlation is more efficient at the bivariate spherical $t_\nu$ distribution for $\nu < 6$.

In the remainder of this section, we examine the influence function of the spatial sign correlation. 
The influence function is based on the notion that estimators are statistical functionals working on distributions. The specific estimate computed from the data set $\X_n$ is then the functional applied to the corresponding empirical distribution. 
We use $\hat{S}$ and $\hat\rho$ to denote the statistical functionals corresponding to the SSCM and the spatial sign correlation, respectively. The influence function $IF(\bx,\hat\rho,F$) describes the effect of an infinitesimal small contamination at point $\bx$ on the functional $\hat\rho$ if the latter is evaluated at distribution $F$. It is an important tool describing the robustness properties of estimators. For a precise definition, interpretation and further details, see, e.g., \citet{Hampel1986} or \cite{Maronna2006}. 
 
 \cite{Croux2010} give the influence function of the off-diagonal element of the SSCM for $p=2$. Calculation of the diagonal elements is straightforward, and we obtain for $F \in \Ee_2(\bmu, V)$:
\[
	IF(\bx,\hat{S},F) = \bx\bx^T/(\bx^T\bx)-S(F).
\] 
Applying the chain rule and using the derivatives calculated in the proof of Proposition \ref{prop:V_0} in the appendix, we arrive at the influence function of the spatial sign correlation. 
 
\begin{proposition} \label{prop:IF} Let $F \in \Ee_2(\bmu, V)$. Then 
	$ IF(\bx,\hat{\rho},F) = $
\[
  \textstyle
	{
	{	-\left\{\left(a^2+1\right)\,\rho\,\sqrt{1-\rho^2}+2\,a\,\rho\, \left(1-\rho^2\right)\right\}\,\left(a^2\,x_2^2+x_1^2\right) - 
		\left\{ \left(  a^4+6\,a^2+1\right)\,\left(\rho^2-1\right)+2\,a\,\left(a^2+1\right)
 		\,\sqrt{1-\rho^2}\,\left(\rho^2-2\right)\right\} \,x_1\,x_2
 	}
 \over{
 		\left\{2 \,a^2\,\sqrt{1-\rho^2}+a\,\left(a^2+1\right)\right\}\,\left(x_2^2+x_1^2 \right)
 	}},
\]
where $\bx = (x_1,x_2)^T$ and $a$ and $\rho$ are as in Proposition \ref{prop:Skor}.
 \end{proposition}
%
The influence function for $a=1$ and $\rho=0$ is illustrated in Figure \ref{figinflu} on the right. It has a discontinuity at the origin and is bounded. Its extreme values $\pm 2$ are attained on the diagonals. Furthermore, $IF(\bx,\hat{\rho},F)$ is bounded in $\bx$ for any fixed values $a$ and $\rho$, but it may get arbitrarily large as $a$ varies. A robustness index that is derived from the influence function is the gross-error sensitivity (GES), defined as
\[
	GES(\hat\rho,F) = \sup_{\bx \in \R^2} \left| IF(\bx,\hat{\rho},F) \right|.
\]
For $a =1$, we obtain
\[
 	\textstyle
	GES(\hat{\rho},F) = 
	{{\left\{\left(\rho^2-1\right)\,\left(-\rho^4+8\,
 \rho^2+4\,\sqrt{1-\rho^2}\,\left(\rho^2-2\right)-8\right)\right\}^{1/2}+|\rho|\,
 \left(\sqrt{1-\rho^2}-\rho^2+1\right)}\over{\sqrt{1-\rho^2}+1}}.
\]
which is depicted in Figure \ref{figinflu} (left).
\citet{CrouxDehon2010} compute the gross-error sensitivities of several nonparametric correlation measures at bivariate normal distributions. Figure \ref{figinflu} (left) corresponds to their Figure 2, complemented by the GES curve of the spatial sign correlation.  The GES is small for any $\rho$, indicating a good robustness against small amounts of outliers. 
We refrain from stating the GES for arbitrary $a$ and $\rho$ explicitly since the formula is rather lengthy. 
\begin{figure}
\begin{center}
\includegraphics[width=0.95\textwidth]{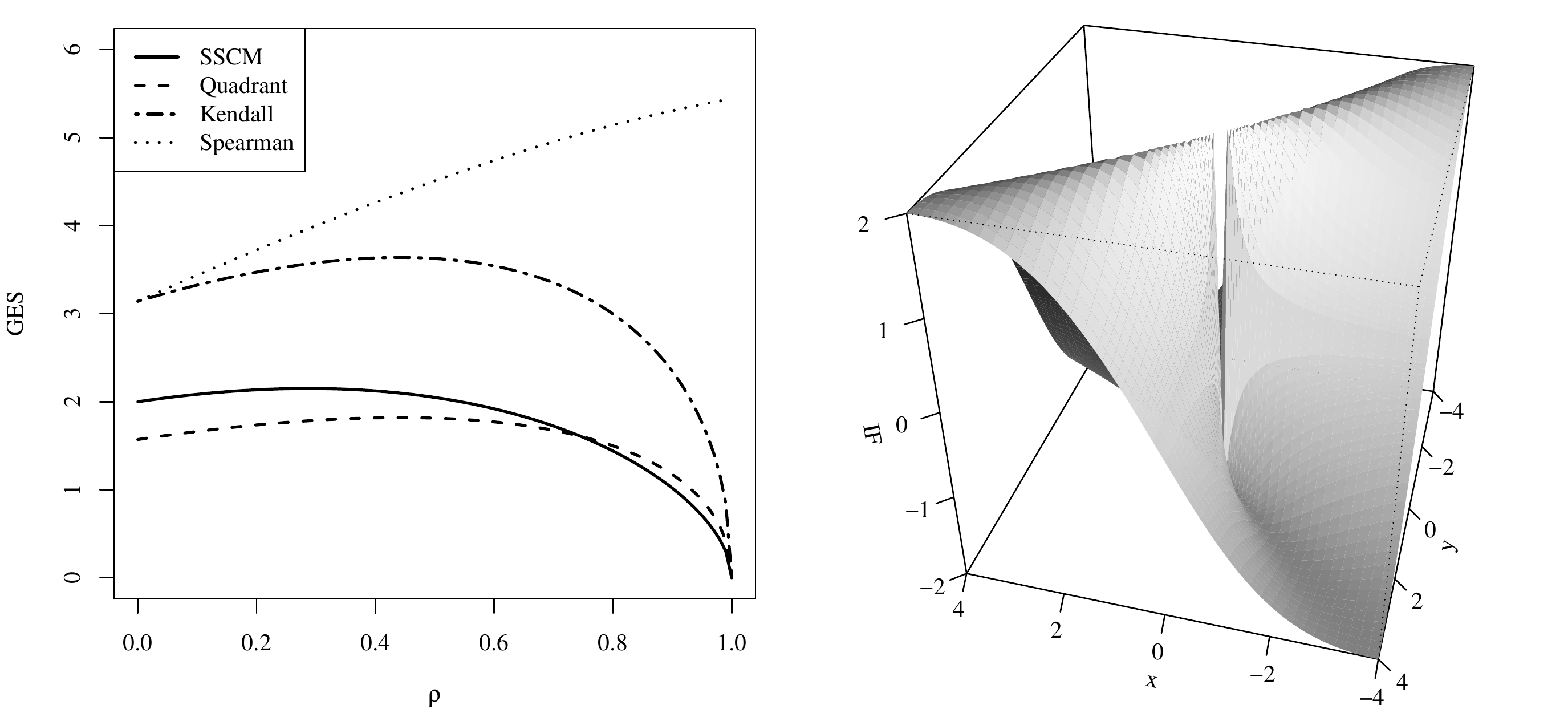}
\end{center}
\vspace{-2.0ex}
\caption{GES for the spatial correlation compared to other popular nonparametric correlation estimators under equal marginal variances (left) and influence function of the spatial correlation for $\rho=0$ and $a=1$ (right).}
\label{figinflu}
\end{figure}



\section{Analytical comparison of robust correlation estimators}
\label{sec:anal.comp}

There are many proposals for robust correlation estimators in the literature. In the second part of this exposition, consisting of Sections \ref{sec:anal.comp} and \ref{sec:num.comp}, we compare the spatial sign correlation $\hat\rho_n$ to a number of prominent alternatives, without claiming or attempting any completeness or ranking. In Section \ref{sec:anal.comp}, we gather the basic analytic results, particularly the asymptotic efficiencies at the normal model, and in Section \ref{sec:num.comp} we compare the finite-sample and robustness properties numerically. 

It is important to note that --- in general --- the estimators mentioned are known to be Fisher-consistent for the correlation only under normality, which often --- as in the case of the spatial sign correlation --- can be relaxed to ellipticity. To put it differently, each of the various correlation estimators\footnote{In this sense, ``correlation'' is understood as monotone dependence.} 
$\hat\theta_n$ estimates some parameter $\theta$ of the bivariate population distribution, which may serve as a measure of monotone dependence, but does in general not coincide with the moment correlation $\rho$. The exact functional connection between $\theta$ and $\rho$ is usually hard to assess for arbitrary distributions, but is known for important subclasses, such as the normal model. If no such function is mentioned in the examples below, it is the identity. 

Let the data be denoted by $\bX_i=(X_i,Y_i)^T$, $i=1,\ldots,n$, independent and normally distributed. Relative efficiencies reported below
are with respect to the sample correlation, which is denoted by $\hat\rho_{\rm Pea}$.
The estimators we will consider can roughly be divided into three groups: We call the first group \emph{nonparametric estimators} since they depend on signs and ranks. Besides the spatial sign correlation, these are the Gaussian rank correlation, Spearman's $\rho$, Kendall's $\tau$, and the quadrant correlation. 
The second group are the \emph{Gnanadesikan-Kettenring-type estimators}, where we consider the $\tau$-scale and the $Q_n$ as scale estimators. We label the third group \emph{affine equivariant estimators}, i.e.\ estimators that are derived from an affine equivariant two-dimensional scatter estimator. Here we consider Tyler's M-estimator, raw and reweighted MCD, the Stahel-Donoho-estimator and the S-estimator with Tukey's biweight-function. The estimators in detail:

\subsection{Nonparametric estimators}
The Gaussian rank correlation is defined as the sample correlation of the normal scores of the data, i.e.\ 
\[
	\hat{\rho}_{\rm GRK}=\frac{1}{c_n} \sum_{i=1}^n \Phi^{-1}\left( \frac{R(X_i)}{n+1}\right)\Phi^{-1}\left(\frac{R(Y_i)}{n+1}\right),
\]
where $c_n=\sum_{i=1}^n\Phi^{-1}\left(\frac{i}{n+1} \right)^2$, $R(X_i)$ is the rank of $X_i$ among $X_1,\ldots, X_n$ and $\Phi^{-1}$ is the quantile function of the standard normal distribution.  The influence function of the Gaussian rank correlation is unbounded, but in finite samples it is much more robust than the Pearson correlation \citep{Boudt2012}. 
Since the Gaussian rank correlation corresponds to the Pearson correlation of transformed data, the pairwise estimation of multidimensional correlation matrices leads always to a non-negative definite estimate.

Another rank based estimator is Spearman's $\rho$, which is the sample correlation of the ranks $R(X_1)$, $\ldots$, $R(X_n)$ and $R(Y_1),\ldots,R(Y_n)$. To obtain a consistent estimator for $\rho$, one has to apply the transformation 
$\hat{\rho}_{\rm Sp.c} = 2\sin\left(\pi \hat{\rho}_{\rm Sp}/6 \right)$,	
which goes back to \citet{pearson:1907}.
Another popular nonparametric estimator is Kendall's $\tau$, which is defined as
\[
	\hat{\rho}_{\rm Ken}=\frac{2}{n(n-1)}\sum_{i > j}\bs\left((X_i-X_j)(Y_i-Y_j)\right),
\]
where $\bs(\cdot)$ is the sign function defined at the beginning of Section \ref{sec:sscm}, here applied to a univariate argument.
It also requires a consistency transformation, which is valid under ellipticity \citep[e.g.][]{Mottonen1999}: 
$\hat{\rho}_{\rm Ken.c}=\sin\left(\pi \hat{\rho}_{\rm Ken}/2\right)$. 
A highly robust, non-parametric procedure based on signs is the quadrant correlation, which can be expressed as
\[
	\hat{\rho}_{\rm Q}=\frac{1}{n}\sum_{i=1}^n\bs\left((X_i-\mbox{med}(X))(Y_i-\mbox{med}(Y))\right),
\]
where $\mbox{med}(X)$ denotes the median of $X_1,\ldots,X_n$. The same transformation $\hat{\rho}_{\rm Q.c}=\sin\left(\pi \hat{\rho}_{\rm Q}/2\right)$ renders this
estimator consistent for $\rho$ under elliptical distributions.
All three nonparametric estimators $\hat{\rho}_{\rm Sp.c}$, $\hat{\rho}_{\rm Ken.c}$, $\hat{\rho}_{\rm Q.c}$ have a bounded influence function and are therefore called B-robust. Their influence functions, asymptotic variances and gross-error sensitivities can be found in \citet{CrouxDehon2010}.

\subsection{GK estimators}
\citet{Gnanadesikan1972} introduced an estimation principle based on robust variance estimation,
\[
	\hat{\rho}=
	\frac{
		\hat\sigma^2(X/\alpha+Y/\beta)-\hat\sigma^2(X/\alpha-Y/\beta)
	}{
		\hat\sigma^2(X/\alpha+Y/\beta)+\hat\sigma^2(X/\alpha-Y/\beta)
	},
\]
where $\hat\sigma$ is can be any robust scale measure and $\alpha=\hat\sigma(X)$, $\beta=\hat\sigma(Y)$. 
Such an estimator can be seen to be Fisher-consistent for $\rho$, regardless of the choice of the scale measure $\hat\sigma$, if $X+Y$, $X-Y$ as well as $X$ and $Y$ have the same distribution up to location and scale, which is fulfilled for elliptical distributions. 
According to \citet{MaGenton2001}, the correlation estimator has the same asymptotic efficiency as the underlying variance estimator. There is also a relationship between the influence functions, which guarantees that the B-robustness translates from the variance to the correlation estimator, see \citet{GentonMa1999}. In the recent literature, there are two proposals for the variance estimation. \citet{Maronnazamar2002} favor the so-called $\tau$-scale:
\[
	\hat{\sigma}_\tau = \frac{\sigma_0^2}{n}\sum_{i=1}^n d_{c_2}\left(\frac{X_i-\hat\mu(X)}{\sigma_0}\right),
	\qquad \mbox{ where } \qquad
	\hat\mu(X)=\frac{\sum_{i=1}^n w_i X_i}{\sum_{i=1}^n w_i},
\]
$w_i = W_{c_1}\!\left\{(X_i-\med(X))/\sigma_0\right\}$, 
$\sigma_0 = \mbox{med}\left\{ \, |X_i-\med{X}| \ : \  i =1,\ldots, n \right\} $, 
$d_{c}(x)=\min(x^2,c^2)$ and 
$W_{c}(x)=\left(1-(x/c)^2\right)^2 \mathds{1}_{\{|x|\leq c\}}$.
They use $c_1=4.5$ and $c_2=3$ to get an efficiency of approximately 0.8 under normality distribution. \citet{MaGenton2001} use the $Q_n$, which is defined as
\[
	Q_n(X)=d\cdot \{|x_i-x_j|:i<j\}_{(k)},
\]
where $k={[n/2]+1 \choose 2}$ and $d$ is a consistency factor equaling $1/(\sqrt{2}\Phi^{-1}(5/8))$ for the normal distribution. This estimator has an efficiency of 0.82, see \citet{RousseeuwCroux1993}. The influence function of the resulting covariance estimator is bounded and can also be found in \citet{MaGenton2001}. 

%
%
%
%
%
%
\subsection{Affine equivariant estimators}
One can estimate the correlation by means of any affine equivariant, bivariate scatter estimator $\hat{V}_n$ using the relation 
	$\hat{\rho}=\hat{v}_{1,2}/\sqrt{\hat{v}_{1,1}\hat{v}_{2,2}}$.
\citet{Taskinen2006} derive the influence function of the correlation estimator from the influence function of $\hat{V}_n$ under elliptical distributions. Furthermore, the asymptotic variance of $\hat{\rho}$ is of the form $(1-\rho^2)^2\cdot ASV(\hat{v}_{1,2}),$ where $ASV(\hat{v}_{1,2})$ is the asymptotic variance of $\hat{v}_{1,2}$ under the corresponding spherical distribution. We consider four examples of robust affine equivariant scatter estimators.

\citet{Tyler1987} proposed an M-estimator for the shape matrix $V$, being a suitably scaled solution of
\[
	\frac{2}{n}\sum_{i=1}^n \frac{(\bX_i - \hat\bmu_n)( \bX_i - \hat\bmu_n)^T}{(\bX_i - \hat\bmu_n)^T \hat{V}^{-1} (\bX_i - \hat\bmu_n)}=\hat{V}_n,
\]
where $\hat\bmu_n$ is a suitable multivariate location estimate. In the simulations in Section \ref{sec:num.comp} we take the spatial median. 
The Tyler estimator can be regarded as an affine equivariant version of the SSCM and is also distribution-free within the elliptical model. The corresponding correlation estimate in two dimensions has an efficiency of 0.5 \citep[e.g.][]{Taskinen2006}. A highly robust, affine equivariant scatter estimator is the minimum covariance determinant estimator (MCD) proposed by \citet{Rousseeuw1985}. For a given trimming constant $\alpha$, it is defined as the sample covariance matrix of the subset of the observations that yields the smallest determinant of the estimated matrix among all subsets of size $\lfloor(1-\alpha)\cdot n\rfloor$.  
Choosing $\alpha=0.5$ results in an asymptotic breakdown point of 0.5. 
Since the asymptotic efficiency, especially in small dimensions, is rather low, the raw MCD is usually followed by a reweighting step. We will call this two-step estimated the weighted MCD. For both, the raw and the weighted MCD, influence functions, consistency factors and asymptotic efficiencies can be found in \citet{CrouxHaesbroeck1999}. 

\citet{Stahel1981} and \citet{Donoho1982} proposed another covariance estimator with an asymptotic breakdown point of 0.5. It is defined as 
\[
	\hat{V}_n =\left(\sum_{i=1}^n w_i\right)^{-1} \sum_{i=1}^n w_i (\bX_i-\hat{\bmu}_n)(\bX_i-\hat{\bmu}_n)^T
	\quad \mbox{where} \quad
	\hat{\bmu}_n= \left(\sum_{i=1}^n w_i\right)^{-1} \sum_{i=1}^nw_i \bX_i,
\]
$ w_i = \min\{ 1 , (c/r_i)^2\}$ and $c$ is often chosen as the 0.95-quantile of the $\chi^2_2$-distribution. The value 
\[
	r_i=\max_{a: |a|=1} \frac{a^T\bX_i-\mbox{med}({a^T\X_n})}{\mbox{MAD}(a^T\X_n)},
\]
is a measure of the outlyingness of $X_i$ among all observations. Any other high-breakdown point location and scale estimators can be used instead of   the median and median absolute deviation \citep[MAD,][]{Hampel1974}. The influence function, asymptotic distribution and gross error sensitivity of the Stahel-Donoho estimator can be found in \citet{Gervini2002}, but an explicit value of its asymptotic efficiency does not seem to be available in the literature.

\citet{Davies1987} proposed a multivariate generalization of S-estimators, being defined as
\[
	(\hat\bmu_n, \hat{V}_n) = \argmin_{\bmu, V} \det(\hat{V}_n) \quad \mbox{ subject to} \quad \ave_{i=1}^n w(\hat{d}_i) = b,
\]
where $\hat{d}_i = \{(\bX_i -\hat\bmu_n)^T \hat{V}_n^{-1}(\bX_i-\hat\mu_n)\}^{1/2}$ and $w$ is a suitable, smooth and bounded, weight function, usually the Tukey--biweight:
\[
	w_c(y)=\min\left(\frac{y^2}{2}-\frac{y^4}{2c^2}+\frac{y^6}{6c^4},\frac{c^2}{6} \right).
\]
Letting $b = E\{ w_c(\| V^{-1/2}(\bX-\bmu)\|) \}$ yields Fisher-consistency at the elliptical population distribution $F$, and if $c$ is chosen such that  $r c^2/ 6= E\{ w_c(\| V^{-1/2}(\bX-\bmu)\|) \}$, the S-estimator achieves an asymptotic breakdown point of $0 < r \le 1/2$. We consider the common standard choices $r = 1/2$ and consistency for $\Sigma$ at the normal model. Asymptotics can be found in \citet{Davies1987}, the influence function was calculated by \citet{Lupuhaa1989}, and efficiencies under normal distribution were calculated for instance in \citet{CrouxHaesbroeck1999}.

Table \ref{tab:1} lists the asymptotic relative efficiencies of the mentioned correlation estimators with respect to the Pearson correlation under normality. 
Specific tuning constants, parameters, weight functions, etc., are chosen as described above.
The efficiency of the nonparametric estimators is declining with $|\rho|$, but the loss is rather small for moderate values.  As we can see, the spatial correlation can well compete with highly robust estimators in terms of efficiency.
\begin{table}
\begin{center}
\begin{tabular}{l|cc||l|cc}
							&    $\rho=0$ & $\rho=0.5 $ & &$\rho=0$ 		& $\rho=0.5$    \\   \hline
Pearson				& \multicolumn{2}{c||}{1} 	&    GK-$Q_n$  	& \multicolumn{2}{c}{0.823}\\
Spatial sign  & 0.5 & 0.464 					  &      GK-$\tau$ 	&\multicolumn{2}{c}{0.8}\\
Gaussian rank &\multicolumn{2}{c||}{1}   	&    Tyler			&\multicolumn{2}{c}{0.5}\\
Spearman			& 0.912 &0.867        		&      rMCD				&\multicolumn{2}{c}{0.033}\\
Kendall				& 0.912 &0.892         		&      wMCD				&\multicolumn{2}{c}{0.401}\\
Quadrant			& 0.405 &0.342            &      S					&\multicolumn{2}{c}{0.377}
\end{tabular}
\caption{Asymptotic efficiency of correlation estimators for $p=2$ under normality}
\label{tab:1}
\end{center}
\end{table}


\section{Numerical comparison}
\label{sec:num.comp}

We compared the correlation estimators in four different situations: under normality, under ellipticity, in outlier scenarios and at non-elliptical distributions. 
We used the statistical software R, including the packages ICSNP (Tyler's M-estimator), MNM (elliptical power exponential distribution), mvtnorm (multivariate normal and elliptical $t$-distributions), pcaPP (spatial median), rrcov (Stahel--Donoho and S-estimator) and robustbase ($\tau$-scale, MCD and $Q_n$). In all scenarios, the estimators were transformed to be Fisher-consistent for the normal distribution. For some estimators, consistency-transformations for other distributions are known as well, but it is unrealistic in practice to know the kind of distribution in advance. Furthermore, the marginal variances are always chosen equal.

\subsection{Results under normality}
First we examine bias and variance under normality. We choose $\rho=0.5$, let the sample size $n$ vary from 5 to 100, and generate 100,000 samples for each sample size. In Figure \ref{fig3} (left), we see that all correlation estimators are biased towards zero in small samples. 
Next to the Pearson correlation, Kendall's $\tau$ and Spearman's $\rho$ (the adequately transformed estimates) are least biased. The bias of the raw MCD still remains heavy even for $n=100$. 
The spatial sign correlation behaves very well in terms of finite-sample variance. On the right-hand side of Figure \ref{fig3}, the variance times $n$ (the ``$n$-stabilized variance'') is plotted against $n$, which is, contrary to most other estimators, nearly a horizontal line.
%
This indicates that
asymptotic tests and confidence intervals based on the spatial sign correlation provide good approximations also in small samples.


\begin{figure}
\begin{center}
\includegraphics[width=0.85\textwidth]{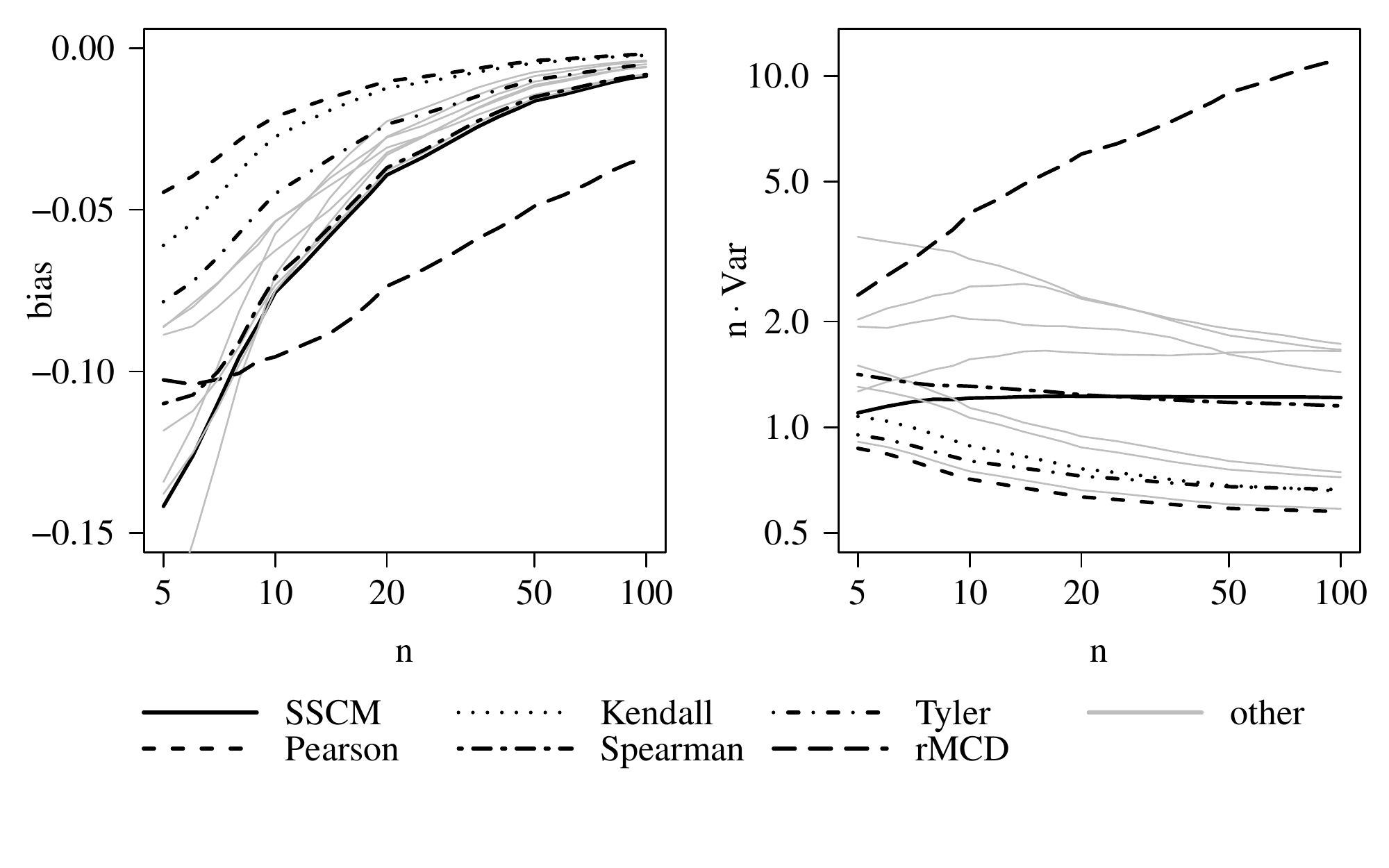}\vspace{-0.7 cm}
\end{center}
\vspace{-2.0ex}
\caption{Simulated finite sample bias (left) and $n\cdot$variance (right) under normality, $\rho=0.5$ and different sample sizes $n$.}
\label{fig3}
\end{figure}

\subsection{Results under elliptical distributions}
Furthermore, the behavior under different elliptical tails is investigated. We consider two subclasses of elliptical distributions that generate varying tails: the $t_\nu$-family 
and the power exponential family \citep[e.g.][pp.~208, 209]{bilodeau:brenner:1999}.
The results for the $t_\nu$ distribution are summarized in Table \ref{tab:2}, where the mean squared error (MSE) of the various correlation estimators based on 100,000 samples are given for $\rho=0$ and $\rho=0.5$ and different degrees of freedom $\nu$. 
Keep in mind that formally the correlation does not exist for one and two degrees of freedom, and we estimate the corresponding parameter $\rho$ of the shape matrix instead, see the remark at the beginning of Section \ref{sec:correlation}.
\begin{table}
\begin{center}
\begin{tabular}{l|cccc|cccc}
$\rho$ 				& \multicolumn{4}{c|}{0}& \multicolumn{4}{c}{0.5}\\
$\nu$						& 1& 2& 5& 10& 1& 2& 5& 10\\\hline\hline
Pearson				& 0.356& 0.112& 0.021& 0.013&  0.283&  0.077&  0.012& 0.007\\\hline
Spatial sign	& 0.020& 0.020& 0.019& 0.020&  0.012&  0.012&  0.012&  0.012\\
Quadrant			& 0.024& 0.024& 0.024& 0.017&  0.017&  0.016&  0.016&  0.016\\
Kendall 			& 0.019& 0.016& 0.013& 0.012&  0.012&  0.010&  0.008&  0.007\\
Spearman			& 0.016& 0.014& 0.012& 0.012&  0.015&  0.011&  0.008&  0.007\\
Gaussian rank	& 0.021& 0.017& 0.013& 0.012&  0.019&  0.013&  0.008&  0.007\\\hline
GK-$Q_n$ 			& 0.021& 0.017& 0.015& 0.014&  0.012&  0.010&  0.009&  0.008\\
GK-$\tau$ 		& 0.024& 0.019& 0.015& 0.014&  0.014&  0.011&  0.009&  0.008\\\hline
Tyler					& 0.020& 0.020& 0.020& 0.020&  0.012&  0.012&  0.012&  0.012\\
rMCD					& 0.076& 0.099& 0.132& 0.149&  0.047&  0.062&  0.085&  0.098\\
wMCD					& 0.037& 0.035& 0.034& 0.032&  0.022&  0.021&  0.020&  0.019\\
S							& 0.033& 0.030& 0.029& 0.029&  0.019&  0.017&  0.017&  0.017\\
Stahel-Donoho	& 0.031& 0.028& 0.026& 0.025&  0.018&  0.016&  0.015&  0.015
\end{tabular}
\caption{MSE under $t_\nu$ distributions with different degrees of freedom and $n=100$.}
\end{center}
\label{tab:2}
\end{table}
The MSE of the spatial sign correlation remains constant with respect to $\nu$, which applies only to the quadrant correlation and Tyler's M-estimator among the other methods. For one degree of freedom and $\rho=0.5$, the spatial sign correlation, together with Kendall's $\tau$ and Tyler's M-estimator, is most efficient. For one degree of freedom and $\rho=0$, Spearman's $\rho$ yields the smallest MSE by far. But this is due to its (asymptotic) bias towards zero. Contrary to Kendall's $\tau$, the consistency transformation applied to Spearman's $\rho$ under normality is not valid under ellipticity in general. 

The MSEs, again based on 100,000 repetitions, for the power exponential distribution are displayed in Figure \ref{fig4}. The sample size is $n = 100$, the true correlation $\rho = 0.5$, and the power parameter $\alpha$ ranges from 0.02 to 2 in 56 (non-equidistant) steps. 
Letting $\alpha = 1$ corresponds to the normal distribution and $\alpha=0.5$ yields the elliptical Laplace distribution. With decreasing $\alpha$, the distribution gets heavier tailed and more peaked in the origin, but all moments exist for any $\alpha > 0$ and the density remains bounded. 
As before, the MSE of the spatial correlation does not depend on the ``tailedness parameter'' $\alpha$, which is in line with the asymptotic result. Only for very small $\alpha$, we observe a slight incline. 
The power exponential distribution with a small $\alpha$ is particularly challenging for robust scatter estimators, since it possesses heavy tails and a probability mass concentration at the origin. Robust estimators downweight or reject outlying observations, which are in this case no contaminations, but carry --- in contrast to the bulk of the data in the center --- the main information about the shape. In fact, the MSE of the raw MCD is above the displayed region in Figure \ref{fig4}. The spatial sign covariance matrix can cope well with such peaked distributions. For $\alpha < 0.1$, we find it, together with Tyler's estimator, to have the smallest MSE. However, it is crucial to use an appropriate location estimator that also works well with peaks at the center, see e.g. the discussion in Section 3.2 of \citet{Duerre2014}. Altogether Kendall's $\tau$ appears to perform best over the whole range of $\alpha$.
\begin{figure}
\begin{center}
\includegraphics[width=0.75\textwidth]{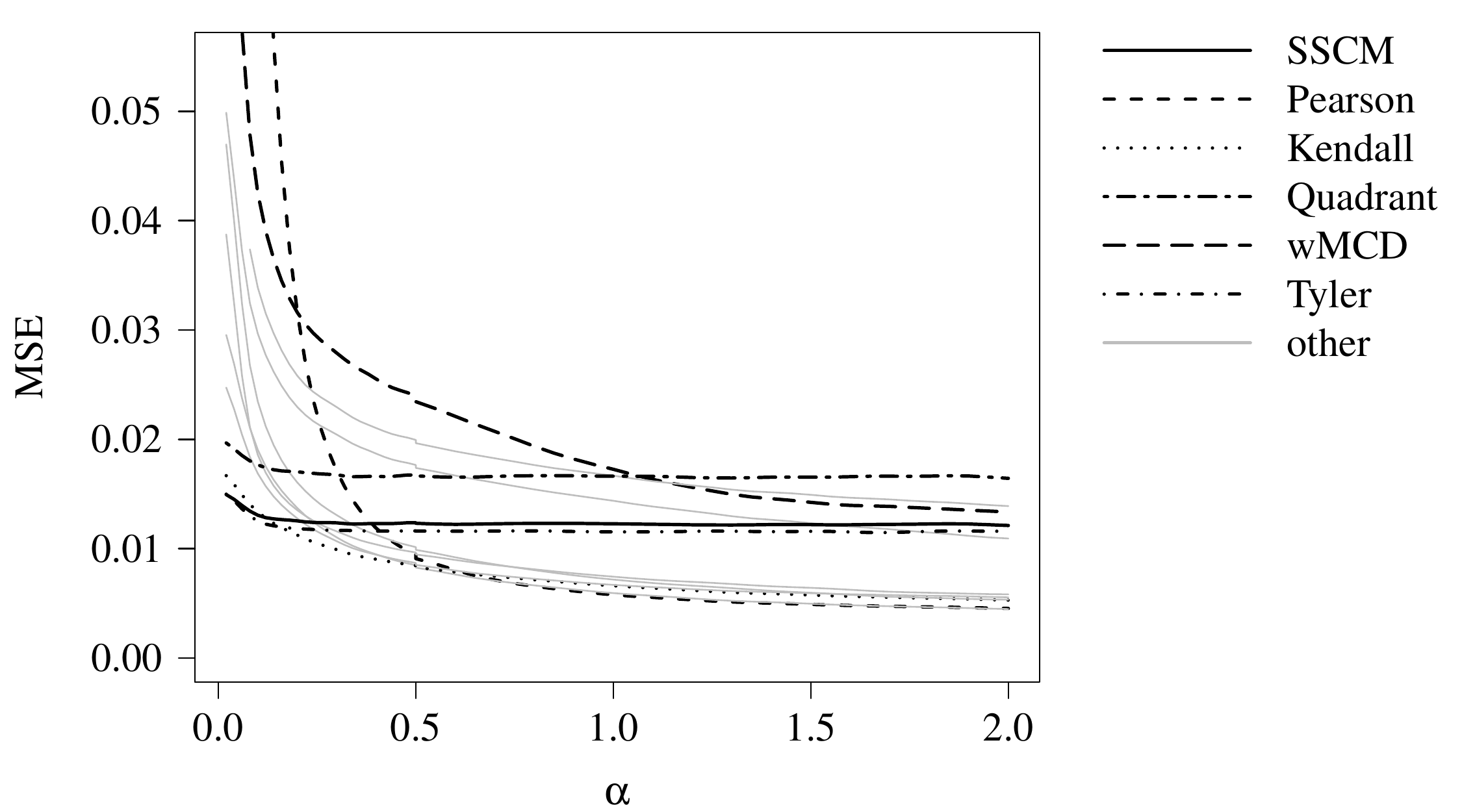}
\end{center}
\vspace{-2.0ex}
\caption{MSE of correlation estimators under the power exponential distribution with different $\alpha$, $\rho=0.5$ and $n=100$.}
\label{fig4}
\end{figure}
\subsection{Results under contamination}
To assess the robustness properties, we consider two scenarios: a single outlier of varying size, and an increasing amount of outliers stemming from a contamination distribution. 
In the first situation, we start from a bivariate normal sample with $\rho=0.5$ and $n=100$ and shift the first observation to the right by a distance $h$ ranging from 0 to 5. This yields a high leverage point, suggesting a smaller correlation. We measure the influence of this one outlier in the $x$-direction by the difference of the estimate before and after the manipulation. The result is a sensitivity curve along the $x$-direction (divided by the factor $n =100$), plotted in Figure \ref{fig5}. The influence of the additive outlier is very small for the spatial sign correlation and also for most other robust estimators. An exception is the Gaussian rank correlation, which is known to have an unbounded influence function. Several highly robust estimators (in particular, the S-estimator and the MCD) 
completely disregard outliers that are sufficiently far away from the bulk of the data, and their sensitivity curves tend back to 0 as $h$ further increases. 
\begin{figure}
\begin{center}
\includegraphics[width=0.75\textwidth]{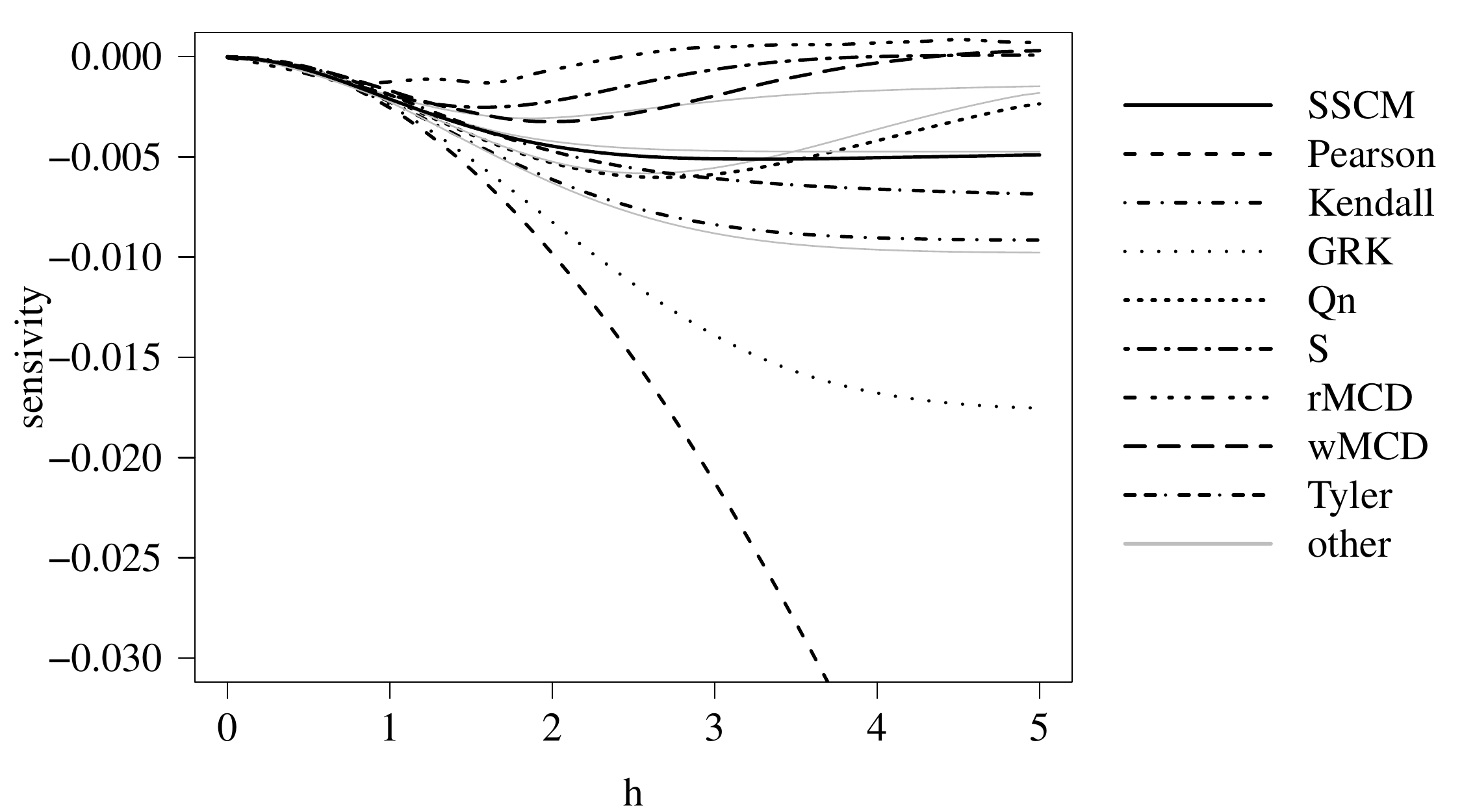}
\end{center}
\vspace{-2.0ex}
\caption{Bias of correlation estimators under normality with $\rho=0.5$, $n=100$ and one additive outlier of size $h$ in the $x$-direction.}
\label{fig5}
\end{figure}

In the second setting, we start as usual with normally distributed data, $\rho=0.5$, marginal variances $1$ and $n=100$. 
Then we replace, one after another, the ``good'' observations by outliers, which stem from a normal distribution with marginal variances 4 and correlation $\rho = -0.5$. In Figure \ref{fig6}, the bias of the estimators (average of 50,000 repetitions) is plotted against the contamination fraction. 
Here the picture is somewhat reversed to the efficiency results under normality: the rather efficient rank-based estimators like Spearman's $\rho$ and Kendall's $\tau$ are substantially biased, and the rather inefficient and highly robust estimators (MCD, S, Stahel-Donoho) perform better. As before, the spatial sign correlation takes a place in the middle.
\begin{figure}
\begin{center}
\includegraphics[width=0.75\textwidth]{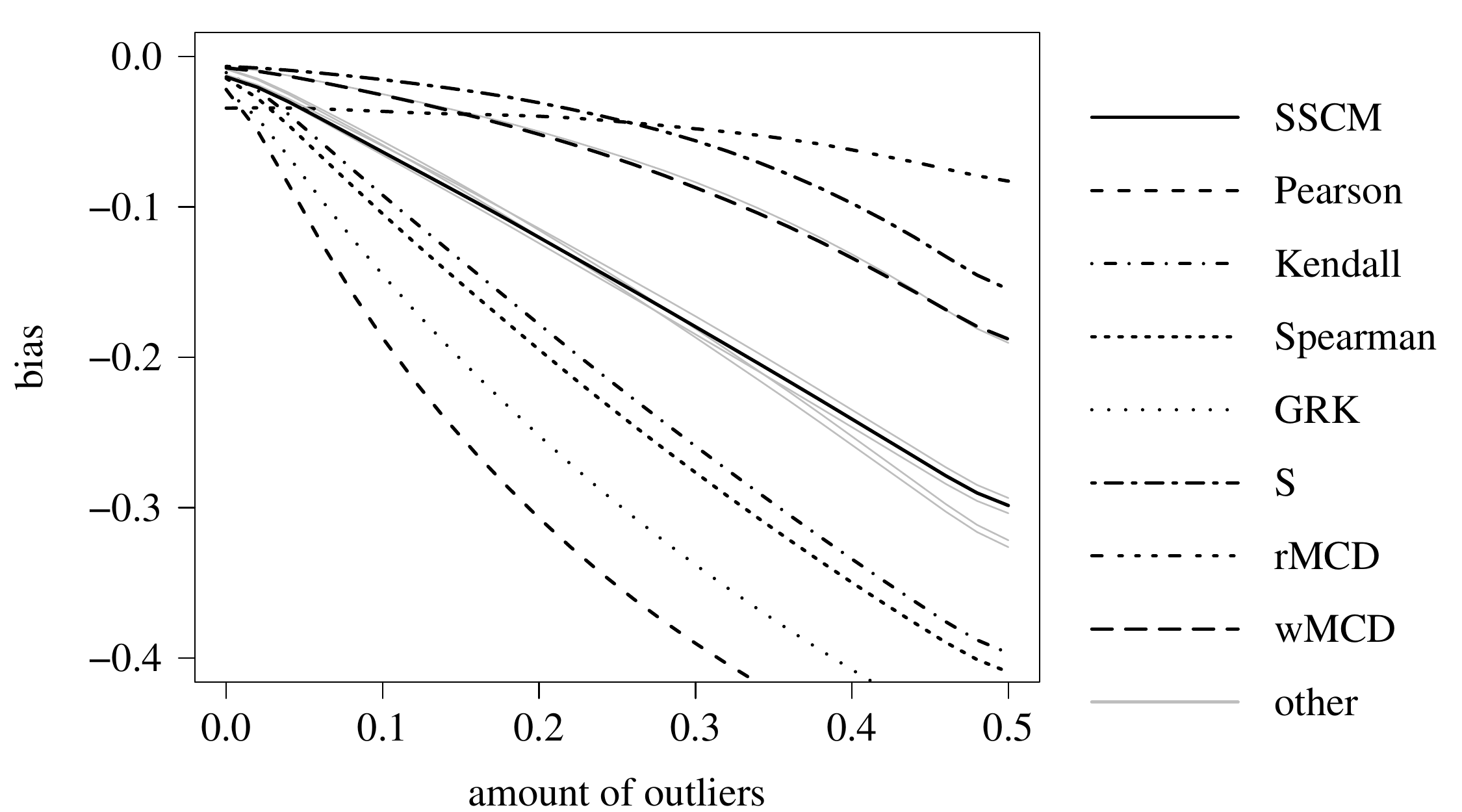}
\end{center}
\vspace{-2.0ex}
\caption{Bias of correlation estimators under normality with $\rho=0.5$, $n=100$ and with a different amount of outliers with correlation $\rho=-0.5$.}
\label{fig6}
\end{figure}

\subsection{Results under non-ellipticity}

The robust correlation estimators are designed to estimate Pearson's moment correlation at the normal model, and the questions remains, what happens in data models that exhibit none of the basic geometric characteristics of the normal distribution, such as symmetry or unimodality. Is the estimate at least somewhere near the actual moment correlation of the population distribution?  An in-depth answer, alone for spatial sign correlation, is beyond the scope of the paper, but we want to get a rough impression in a simulated example. We consider a unimodal, but heavily skewed distribution.
Let $X = \alpha Z_1 + Z_2$ and $ Y = Z_1+\alpha Z_2$, where $\alpha$ is a scalar parameter and $Z_1$ and $Z_2$ are two independent, exponentially  distributed random variables (with parameter $\lambda = 1$). By letting $\alpha$ vary between 0 and 1, one can generate any (positive) correlation $\rho$ between $X$ and $Y$. The explicit formula is
\[ 
	\alpha = (1-\sqrt{1-\rho^2})/\rho
\] 
In Figure \ref{fig7}, the bias of the estimators (based 50000 repetitions) is plotted against the correlation $\rho$. The sample size is $n = 100$.
It is not surprising that most estimators, particularly the nonparametric ones including the spatial sign correlation, are substantially biased. Besides the sample correlation, we find the Gnanadesikan-Kettenring-estimator based on the $Q_n$ (but not on the $\tau$-scale) and the $S$-estimator to be nearly unbiased. 
\begin{figure}
\begin{center}
\includegraphics[width=0.75\textwidth]{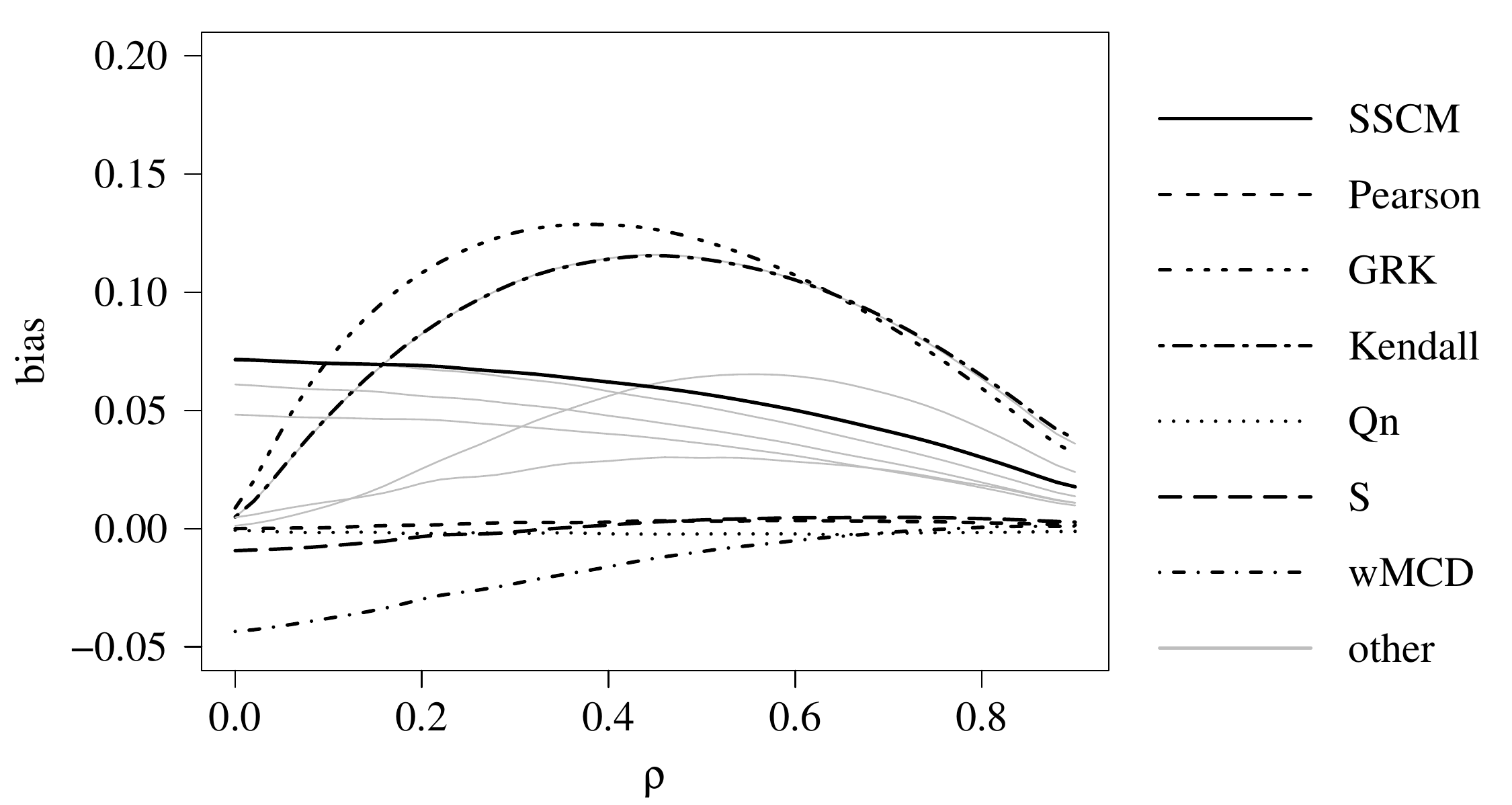}
\end{center}
\vspace{-2.0ex}
\caption{Bias of correlation estimators at a bivariate exponential distribution for $n=100$ and different $\rho$.}
\label{fig7}
\end{figure}

\appendix

\section{Proofs}

\begin{proof}[Proof of Proposition \ref{prop:S}] 
The proofs of parts (\ref{S 1}) and (\ref{S 2}) are fairly straightforward employing the definitions of $\bmu(\bX)$ and $S(\bX)$. The key is the orthogonal equivariance of the spatial median and the orthogonal invariance of the spatial sign. A proof of a more general version of part (\ref{S 2}) can also be found in \citet{Visuri2001}. It remains to show part (\ref{S 3}). We only consider the non-trivial case $\lambda_1 \neq \lambda_2$. Since $\bX \sim F \in \Ee_p(\bmu,V)$, there exists a spherical random variable $\bY$ such that $\bX = U\Lambda^\frac{1}{2}\bY + \bmu$, with $U$ and $\Lambda$ as in (\ref{evd}).
%
%
%
We thus have
\[
S(\bX) 
= \mathbb{E}\left\{\frac{(\bX-\bmu)(\bX-\bmu)^T}{(\bX-\bmu)^T(\bX-\bmu)}\right\}
= U\, \mathbb{E} \left\{\frac{\Lambda^{1/2} \bY\bY^T \Lambda^{1/2}}{ \bY^T\Lambda\bY}\right\}U^T.
\]
It remains to evaluate the diagonal elements $\delta_1$ and $\delta_2$ of the expectation on the right-hand side. (Since spherical distributions have symmetrically distributed margins, the off-diagonal elements are zero.) The spatial sign is distribution-free within the elliptical model, i.e.\ $\bs(\bX) = \bs(\tilde\bX)$ in distribution for any two elliptical random vectors $\bX$ and $\tilde\bX$ with the same shape matrix $V$.  The distribution of $\bs(\bX)$ for elliptical $\bX$ is also known as the \emph{angular central Gaussian distribution}, cf.\ \citet{Tyler1987b}. Hence we can choose any spherical distribution for $\bY$, e.g.\ the uniform distribution on the unit circle with density
\begin{align} \label{Formel2}
	f(\by)=\frac{1}{\pi}\mathds{1}_{[0,1]}(\by^T \by),
\end{align}
which yields with $\by = (y_1, y_2)$
\begin{align*}
\delta_1&=\frac{1}{\pi}\int_0^1\int_{-\sqrt{1-y_1^2}}^{\sqrt{1-y_1^2}}\frac{\lambda_1y_1^2}{\lambda_1y_1^2+\lambda_2y_2^2}dy_2dy_1.
\end{align*}
Substituting spherical coordinates $y_1 = r \cos(\alpha)$, $y_2 = r \sin(\alpha)$, we obtain 
\begin{align*}
\delta_1=\frac{1}{\pi}\int_0^1\int_0^{2\pi}\frac{\lambda_1r^3\cos(\alpha)^2}{\lambda_1r^2\cos(\alpha)^2+\lambda_2r^2\sin(\alpha)^2}d\alpha dr.
\end{align*}
Using the identities $\cos(\alpha)=(e^{i\alpha}+e^{-i\alpha})/2$ and $\sin(\alpha)= (e^{i\alpha}-e^{i\alpha})/(2i)$, we further substitute $z = e^{i\alpha}$ and get
\begin{align}
\delta_1= \frac{1}{\pi}\int_0^1 r  \oint_{\Gamma} \frac{\lambda_1(z^2+1)^2}{iz((\lambda_1-\lambda_2)z^4+2(\lambda_1+\lambda_2)z^2+(\lambda_1-\lambda_2))}dz \, dr,
\label{Formel1}
\end{align} 
where $\Gamma$ denotes the unit circle on the complex plane, and $\oint_\Gamma$ the (closed curve) line integral along $\Gamma$. We apply the residue theorem to solve the inner line integral \citep[e.g.][p.~149]{Ahlfors1966}.
The integrand is meromorphic and has no pole on the boundary of the unit circle. The residue theorem thus yields
\begin{align*}
\delta_1= \frac{1}{\pi}\int_0^1 2\pi i r\sum_{a\in P}\mbox{Res}(\phi,a)dr,
\end{align*}
where $\phi$ is the integrand of equation (\ref{Formel1}), $P$ its set of poles within the unit circle, and $\mbox{Res}(\phi,a)$ the residue of $\phi$ in $a$. The integrand $\phi$ has three poles inside the unit circle
\be \label{eq:poles} 
z_1=0, \quad z_{2/3}=\pm \frac{\sqrt{\lambda_1}-\sqrt{\lambda_2}}{\sqrt{\lambda_1-\lambda_2}}
\ee
with residues
\begin{align*}
	\mbox{Res}(\phi,0)=\frac{-i\lambda_1}{(\lambda_1-\lambda_2)}, \quad 
	\mbox{Res}(\phi,z_{2/3})=\frac{i\sqrt{\lambda_1\lambda_2}}{2(\lambda_1-\lambda_2)}.
\end{align*}
Hence we obtain
\begin{align*}
\delta_1=\frac{1}{\pi}\int_0^1 2\pi r \frac{\sqrt{\lambda_1}}{\sqrt{\lambda_1}+\sqrt{\lambda_2}}dr = \frac{\sqrt{\lambda_1}}{\sqrt{\lambda_1}+\sqrt{\lambda_2}}.
\end{align*}
The expression for $\delta_2$ is obtained by exchanging $\lambda_1$ and $\lambda_2$. The proof is complete.
\end{proof}

\begin{proof}[Proof of Proposition \ref{prop:S.hat}]
Parts (\ref{S.hat 1}) and (\ref{S.hat 2}) are proved in \citet{Duerre2014}. In particular, Theorem 3 of \citet{Duerre2014} identifies $W_S$, under the conditions of part (\ref{S.hat 2}), as the asymptotic covariance matrix of the SSCM with known location, i.e., 
\be \label{eq:W_S}
	W_S = \Cov\left( \vec \left\{ \bs(\bX-\bmu) \bs(\bX-\bmu)^T \right\} \right).
\ee
For proving part (\ref{S.hat 3}) it remains to evaluate (\ref{eq:W_S}) under the additional assumption $\bX \sim F \in \Ee_2(\bmu,V)$. As in the proof of Proposition \ref{prop:S}, we use the representation $\bX = U \Lambda^{1/2} \bY + \bmu$, where $\bY = (Y_1,Y_2)^T$ is a spherical random vector. Again, we only consider the non-trivial case $\lambda_1\neq \lambda_2$ and obtain
\be \label{eq:W_S2}
	W_S = (U\otimes U) \,\Cov \left\{ \vec\left( \frac{\Lambda^{1/2} \bY\bY^T \Lambda^{1/2}}{\bY^T\Lambda\bY}\right)\right\} (U\otimes U)^T
\ee
The inner matrix on the right-hand side is
{\small
\[
	\mathbb{E}\left\{
		\frac{1}{(\bY^T\Lambda\bY)^2}
		\begin{pmatrix}
			\lambda_1^2Y_1^4&0&0&\lambda_1Y_1^2\lambda_2Y_2^2\nonumber\\
			0&\lambda_1Y_1^2\lambda_2Y_2^2&\lambda_1Y_1^2\lambda_2Y_2^2&0\\
			0&\lambda_1Y_1^2\lambda_2Y_2^2&\lambda_1Y_1^2\lambda_2Y_2^2&0\\
			\lambda_1Y_1^2\lambda_2Y_2^2&0&0&\lambda_2^2Y_2^4
		\end{pmatrix}
		\right\}
\]\be \label{eq:matrix}
	\qquad \qquad -
	\begin{pmatrix}
		\delta_1^2&0&0&\delta_1\delta_2\\
		0&0&0&0\\
		0&0&0&0\\
		\delta_1\delta_2&0&0&\delta_2^2
	\end{pmatrix}
\ee}
It remains to solve the three integrals
{\small
\[
I_1 = \mathbb{E}\left\{\! \frac{\lambda_1^2 Y_1^4}{(\lambda_1Y_1^2 \!+\!\lambda_2Y_2^2)^2}\!\right\},\,
I_2 = \mathbb{E}\left\{\! \frac{\lambda_2^2 Y_2^4}{(\lambda_1Y_1^2 \!+\!\lambda_2Y_2^2)^2}\!\right\}, \, 
I_3 = \mathbb{E}\left\{\! \frac{\lambda_1 Y_1^2\lambda_2Y_2^2}{(\lambda_1Y_1^2\!+\!\lambda_2Y_2^2)^2}\!\right\},
\]}

\vspace{-2.0ex}
\noindent
where $I_1$ and $I_2$ are of the same type: $I_2$ is obtained from $I_1$ by simply exchanging $\lambda_1$ and $\lambda_2$. We start with $I_1$. By a fully analogous chain of arguments and manipulations as in the proof of Proposition \ref{prop:S}, we arrive at
\[
	I_1 = 
   \frac{1}{\pi}\int_0^1 r \oint_\Gamma
   \frac{\lambda_1^2(z^2+1)^4}{iz((\lambda_1-\lambda_2)z^4+2(\lambda_1+\lambda_2)z^2+(\lambda_1-\lambda_2))^2}
   dz \, dr.
\]
and apply again the residue theorem to solve the inner line integral.
We call the integrand $\phi_1$ and observe that it has the same singularities as the integrand $\phi$ in the proof of Proposition \ref{prop:S}, cf.~(\ref{eq:poles}). However, the poles $z_2$ and $z_3$ are of order two here, resulting in the residues
\begin{align*}
	\mbox{Res}(\phi_1,0) = \frac{-i\lambda_1^2}{(\lambda_1-\lambda_2)^2}, \qquad 
	\mbox{Res}(\phi_2,z_{2/3}) = \frac{i\sqrt{\lambda_1\lambda_2}(3\lambda_1-\lambda_2)}{4(\lambda_1-\lambda_2)^2}.
\end{align*}
Hence we obtain
\[
	I_1 = 
	\frac{1}{\pi}\int_0^1 2\pi r \frac{\lambda_1^2-\frac{1}{2}\sqrt{\lambda_1\lambda_2}(3\lambda_1-\lambda_2)}{(\lambda_1-\lambda_2)^2}dr
	\ = \
	\frac{\lambda_1^2-\frac{1}{2}\sqrt{\lambda_1\lambda_2}(3\lambda_1-\lambda_2)}{(\lambda_1-\lambda_2)^2}.
\]
It remains to solve $I_3$, which we transform, again, by the same chain of arguments as in the proof of Proposition \ref{prop:S}, to
\begin{align*}
I_3 \, = \, 
	- \frac{1}{\pi}
	\int_0^1 r \oint_\Gamma
		\frac{\lambda_1\lambda_2(z^4-1)^2}{iz((\lambda_1-\lambda_2)z^4+2(\lambda_1+\lambda_2)z^2+(\lambda_1-\lambda_2))^2}
	dz \, dr.
\end{align*}
We call the integrand $\phi_3$. Its poles are also given by (\ref{eq:poles}) with $z_2$ and $z_3$ being of order two, resulting in the residues
\begin{align*}
\mbox{Res}(\phi_3,0) = \frac{ - i \lambda_1\lambda_2}{(\lambda_1-\lambda_2)^2}, \qquad
\mbox{Res}(\phi_3,z_{2/3}) = \frac{i \sqrt{\lambda_1\lambda_2}(\lambda_1+\lambda_2)}{4(\lambda_1-\lambda_2)^2}.
\end{align*}
We finally arrive at
\begin{align*}
I_3 =
-\frac{1}{\pi}\int_0^1 2\pi r \frac{\lambda_1\lambda_2-\frac{1}{2}\sqrt{\lambda_1\lambda_2}(\lambda_1+\lambda_2)}{(\lambda_1-\lambda_2)^2}
=
	\frac{-\lambda_1\lambda_2+\frac{1}{2}\sqrt{\lambda_1\lambda_2}(\lambda_1+\lambda_2)}{(\lambda_1-\lambda_2)^2}.
\end{align*}
Plugging the obtained expressions for $I_1$, $I_2$ and $I_3$ into the matrix (\ref{eq:matrix}) and observing (\ref{eq:W_S2}) yields the expression for $W_S$ given in Proposition \ref{prop:S.hat} (\ref{S.hat 3}). The proof is complete.
\end{proof}

Towards the proof of Propostion \ref{prop:Skor}, we consider as an intermediate step  the SSCM-based shape estimator $\hat{V}_n$ defined at the beginning of Section \ref{sec:correlation}. Precisely, we give the asymptotic distribution of the estimator 
\[
	\hat{V}_{0,n} = 	
	\begin{pmatrix}
		\hat{v}_{0,11} & \hat{v}_{0,12} \\
		\hat{v}_{0,12} & \hat{v}_{0,22} \\
		\end{pmatrix}
	= 
	\frac{1}{\sqrt{ \hat{v}_{11} \hat{v}_{22} }} \hat{V}_n. 	
\]
We have remarked at the end of Section \ref{sec:intro} that, for analyzing the scale-invariant properties of the shape of an elliptical distribution, fixing the overall scale of the shape matrix $V$ is not necessary, and we view the shape as an equivalence class of positive definite matrices being proportional to each other. 
For explicit computations, however, it is at some point necessary to fix the scale, that is, picking one specific representative from the equivalence class. 
Various ways of standardizing the shape 
can be found in the literature. \citet{Paindaveine2008} argues to choose $\det(V) = 1$, which corresponds to our choice of $\hat{V}_n$ in Section \ref{sec:correlation}. However, for our purposes, it is most convenient to standardize $V$ such that the product of its diagonal elements is 1, which corresponds to $\hVnn$ described above. Accordingly, we denote by $V_0$ the representative of the equivalence class with this property (reciprocal diagonal elements) and parametrize it as
\be
	V_0 =
	\begin{pmatrix}
		a & \rho \\
		\rho & a^{-1} \\
		\end{pmatrix}, 		\label{eq:V_0}
\ee
where the parameters $a$ and $\rho$ have the same meaning as in Section \ref{sec:correlation}, that is, the ratio of the diagonal elements of $V$ and the correlation, respectively. Lengthy but straightforward calculus yields
\be
	\hat{v}_{0,12} = \frac{c\hs_{12} b}{\sqrt{(\hs_{12}^2+b^2)^2+(\hs_{12}cb)^2}}, \label{Formel4}
\ee\be
	\hat{v}_{0,11} = \frac{2\bs(\hs_{12})\sqrt{-\rho(4\rho\hs_{12}^2+4\sqrt{1-\ts_{12}^2}\hs_{12}-\ts_{12})}+2\ts_{1,2}-4\sqrt{1-\ts_{12}^2}\hs_{12}}{4\hs_{12 b}} \label{Formel5}
\ee
where, as before, $\hat{s}_{ij}$, $i,j = 1,2$, denote the elements of the SSCM $\hat{S}(\X_n;\bmu_n)$, and $b$, $c$ and $d$ are defined in (\ref{eq:cdb}).
The following proposition summarizes the asymptotic behavior of the estimator $\hVnn$.
\begin{proposition} \label{prop:V_0}
Under the assumptions of Proposition \ref{prop:Skor}, we have for $n \to \infty$ that
\begin{enumerate}[(1)]
\item \label{prop:V_0 1}
	$\hVnn \asc  V_\rho$  and
\item \label{prop:V_0 2}
	$\sqrt{n} \left\{(\hat{v}_{0,11},\hat{v}_{0,12})^T - (a,\rho)^T \right\} \cid N_{2}\left(\bNull, W_{V_0}\right)$, where $W_{V_0} = G W_S G^T$ \\[0.5ex] with
\begin{align*}
	G & =\frac{\left((a^2+1)\,\sqrt{1-\rho^2}+2a(1- \rho^2)\right)}{\sqrt{1-\rho^2}\left(4a^2\rho^2+(a^2-1)^2\right)}
	\begin{pmatrix}
		g_{1,1} & g_{1,2} & 0 & 0 \\
		g_{2,1} & g_{2,2} & 0 & 0 
	\end{pmatrix}
\end{align*}
and
\begin{align*}
	g_{1,1}&= (a^2-1)^2\,\sqrt{1-\rho^2}+2\,a\,  (a^2+1)\,\rho^2, \\
	g_{1,2}&= (a-1)(a+1)\,\rho\,\left\{2\,a\,\sqrt{1-\rho^2}-a^2-1\right\},\\
	g_{2,1}&= \frac{1}{a}\left\{(a^2+1)\,\sqrt{1-\rho^2}-2a(1- \rho^2)\right\},\\
	g_{2,2}&= 2(a^2+1)\rho^2\sqrt{1-\rho^2} \, + \, a^{-1} (a^2-1)^2(1-\rho^2).
\end{align*} 
\end{enumerate}
\end{proposition}
\begin{proof}[Proof of Proposition \ref{prop:V_0}.]
Part (1) is a consequence of the continuous mapping theorem, part (2) follows with the delta method. 
Note that $V_0$ is specified by the two elements $\hat{v}_{0,11}$ and $\hat{v}_{0,12}$, and, likewise, $\hat{S}_n = \hat{S}_n(\X_n;\bmu_n)$ by the two elements $\hat{s}_{11}$ and $\hat{s}_{12}$. Let $H$ be the function that maps $(\hat{s}_{11}, \hat{s}_{12})$ to $(\hat{v}_{0,11}, \hat{v}_{0,12})$
and $(s_{11},s_{12})$ to $(a,\rho)$. It is given explicitly by the formulas (\ref{Formel4}) and (\ref{Formel5}), from which we can compute its derivative. However, due to the complex structure of $H$, it is a cumbersome task to compute its derivative. It is much easier to compute the derivative of its inverse and apply the inverse function theorem. With $\{ (x,y) \,|\, 0 < x , |y| < x \}$ and $\{ (x,y) \,|\, 0 < x < 1, |y| < x \}$ being its domain and image, respectively, the function $H$ is invertible and continuously differentiable. Let $J$ denote its inverse. 
The function $J$ maps $(a,\rho)$ to $(s_{11},s_{12})$ and 
is described in Proposition \ref{prop:S}. 
In the following, we will compute its derivate, for which we require an explicit form of $J$.  
The eigenvalue decomposition of $V_0$ is given by
\begin{align*}
		\lambda_{1/2}= (2a)^{-1} \left( a^2 + 1 \pm \sqrt{q}  \right)
\end{align*} 
and 
\begin{align*}
	U = 
	\begin{pmatrix}
		{{2\,a\,\left| \rho\right| } \over{\left\{ \left(\sqrt{q}-a^2+1\right)^2+\,4 a^2\rho^2\right\}^{1/2}}} &
		{{2\,a\,\left| \rho\right| } \over{\left\{ \left(\sqrt{q}+a^2-1\right)^2 + \,4a^2\rho^2\right\}^{1/2}}} \\[3ex]
	 {{{\it \bs}\left(\rho\right)\,\left(\sqrt{q}-a^2+1\right)} \over{ \left\{ \left(\sqrt{q}-a^2+1\right)^2 + \,4 a^2\rho^2\right\}^{1/2}}} &
	 -{{{\it \bs}\left(\rho\right)\,\left(\sqrt{q}+a^2-1\right)} \over{ \left\{ \left(\sqrt{q}+a^2-1\right)^2 + \,4a^2\rho^2\right\}^{1/2}}}  
  \end{pmatrix},
\end{align*}
where $q=4a^2\rho^2+(a^2-1)^2$. By Proposition \ref{prop:S} (\ref{S 1}) and (\ref{S 2}) we find 
{\small
\begin{align*}
s_{11}=\frac{\sqrt{k}\{ 4a^2\rho^2+\sqrt{q}(a^2-1)+(a^2-1)^2\} + 
						 \sqrt{m}\{ 4a^2\rho^2+\sqrt{q}(1-a^2)+(a^2-1)^2)\} }
				{2 q (\sqrt{m}+\sqrt{k})}
\end{align*}}

\vspace{-1ex}
\noindent and $s_{12}=(2 q)^{-1} a\rho(\sqrt{k}-\sqrt{m})^2$, where $k = a^2+1+\sqrt{q}$ and $m = a^2+1-\sqrt{q}$. 
%
The derivative of $J$ is
\begin{align*}\dsD J(a,\rho)=
\begin{pmatrix}
   \frac{2a(a^2+1)\rho^2\sqrt{1-\rho^2}+(a^2-1)(1-\rho^2)  	}{  	q((a^2+1)\sqrt{1-\rho^2}+2a(1-\rho^2))}  &
		-\frac{(a-1)a(a+1)\rho(2a\sqrt{1-\rho^2}-a^2-1)        	}{		q((a^2+1)\sqrt{1-\rho^2}+2a(1-\rho^2))}  \\
		-\frac{(a-1)(a+1)\rho((a^2+1)\sqrt{1-\rho^2}-2a(1-\rho^2))}{	q((a^2+1)\sqrt{1-\rho^2}+2a(1-\rho^2))}	 &
		\frac{a((a-1)^2\sqrt{1-\rho^2}+2a(a^2+1)\rho^2)					}{		q((a^2+1)\sqrt{1-\rho^2}+2a(1-\rho^2))}
\end{pmatrix}.
\end{align*}
The determinant of this matrix is 
\[
	\det\dsD J(a,\rho) = a\sqrt{1-\rho^2} \left\{\left(a^2+1\right)\,\sqrt{1-\rho^2 }+2\,a\,\left(1-\rho^2\right)\right\}^{-2}.
\]
By virtue of the inverse function theorem, we have $\dsD H(s_{11},s_{12}) = (\dsD J (a,\rho))^{-1}$. Hence we obtain $\dsD H(s_{11},s_{12})$ by inverting the $2\times 2$ matrix $\dsD J(a,\rho)$. It can be seen to be (except for the zero columns) the matrix $G$ in Proposition \ref{prop:V_0}.
The proof is complete.
\end{proof}

\begin{proof}[Proof of Proposition \ref{prop:Skor}.]
Proposition \ref{prop:Skor} is an immediate corollary of Proposition \ref{prop:V_0}, noting that $\hat\rho_n = \hat{v}_{0,12}$. The asymptotic variance of $\hat\rho_n$ is the lower diagonal element of $W_{V_0}$ given in Proposition \ref{prop:V_0}. 
\end{proof}

\bibliographystyle{abbrvnat}
{\small

\begin{thebibliography}{44}
\providecommand{\natexlab}[1]{#1}
\providecommand{\url}[1]{\texttt{#1}}
\expandafter\ifx\csname urlstyle\endcsname\relax
  \providecommand{\doi}[1]{doi: #1}\else
  \providecommand{\doi}{doi: \begingroup \urlstyle{rm}\Url}\fi

\bibitem[Ahlfors(1966)]{Ahlfors1966}
L.~V. Ahlfors.
\newblock \emph{Complex analysis}.
\newblock New York: McGraw-Hill, 2nd edition, 1966.

\bibitem[Bali et~al.(2011)Bali, Boente, Tyler, and Wang]{Bali2011}
J.~L. Bali, G.~Boente, D.~E. Tyler, and J.-L. Wang.
\newblock Robust functional principal components: A projection-pursuit
  approach.
\newblock \emph{The Annals of Statistics}, 39\penalty0 (6):\penalty0
  2852--2882, 12 2011.

\bibitem[Bilodeau and Brenner(1999)]{bilodeau:brenner:1999}
M.~Bilodeau and D.~Brenner.
\newblock \emph{Theory of Multivariate Statistics.}
\newblock Springer Texts in Statistics. {New York: Springer}, 1999.

\bibitem[Boudt et~al.(2012)Boudt, Cornelissen, and Croux]{Boudt2012}
K.~Boudt, J.~Cornelissen, and C.~Croux.
\newblock The {G}aussian rank correlation estimator: robustness properties.
\newblock \emph{Statistics and Computing}, 22\penalty0 (2):\penalty0 471--483,
  2012.

\bibitem[Croux and Dehon(2010)]{CrouxDehon2010}
C.~Croux and C.~Dehon.
\newblock Influence functions of the {S}pearman and {K}endall correlation
  measures.
\newblock \emph{Statistical Methods \& Applications}, 19\penalty0 (4):\penalty0
  497--515, 2010.

\bibitem[Croux and Haesbroeck(1999)]{CrouxHaesbroeck1999}
C.~Croux and G.~Haesbroeck.
\newblock Influence function and efficiency of the minimum covariance
  determinant scatter matrix estimator.
\newblock \emph{Journal of Multivariate Analysis}, 71\penalty0 (2):\penalty0
  161--190, 1999.

\bibitem[Croux et~al.(2002)Croux, Ollila, and Oja]{Croux2002}
C.~Croux, E.~Ollila, and H.~Oja.
\newblock Sign and rank covariance matrices: statistical properties and
  application to principal components analysis.
\newblock In \emph{Statistical data analysis based on the $L_1$-norm and
  related methods}, pages 257--269. Springer, 2002.

\bibitem[Croux et~al.(2010)Croux, Dehon, and Yadine]{Croux2010}
C.~Croux, C.~Dehon, and A.~Yadine.
\newblock The $k$-step spatial sign covariance matrix.
\newblock \emph{Advances in data analysis and classification}, 4\penalty0
  (2-3):\penalty0 137--150, 2010.

\bibitem[Davies(1987)]{Davies1987}
P.~Davies.
\newblock Asymptotic behaviour of {S}-estimates of multivariate location
  parameters and dispersion matrices.
\newblock \emph{The Annals of Statistics}, pages 1269--1292, 1987.

\bibitem[Davies and Gather(2005)]{Davies2005}
P.~L. Davies and U.~Gather.
\newblock {Breakdown and groups.}
\newblock \emph{The Annals of Statistics}, 33\penalty0 (3):\penalty0 977--1035,
  2005.

\bibitem[Donoho(1982)]{Donoho1982}
D.~L. Donoho.
\newblock \emph{Breakdown properties of multivariate location estimators}.
\newblock PhD thesis, Harvard University, 1982.

\bibitem[D{\"u}rre et~al.(2014)D{\"u}rre, Vogel, and Tyler]{Duerre2014}
A.~D{\"u}rre, D.~Vogel, and D.~E. Tyler.
\newblock The spatial sign covariance matrix with unknown location.
\newblock \emph{arXiv:1307.5706}, under revision for Journal of Multivariate
  Analysis, 2014.

\bibitem[Genton and Ma(1999)]{GentonMa1999}
M.~G. Genton and Y.~Ma.
\newblock Robustness properties of dispersion estimators.
\newblock \emph{Statistics \& Probability Letters}, 44\penalty0 (4):\penalty0
  343--350, 1999.

\bibitem[Gervini(2002)]{Gervini2002}
D.~Gervini.
\newblock The influence function of the {S}tahel--{D}onoho estimator of
  multivariate location and scatter.
\newblock \emph{Statistics \& Probability Letters}, 60\penalty0 (4):\penalty0
  425--435, 2002.

\bibitem[Gervini(2008)]{Gervini2008}
D.~Gervini.
\newblock Robust functional estimation using the median and spherical principal
  components.
\newblock \emph{Biometrika}, 95\penalty0 (3):\penalty0 587--600, 2008.

\bibitem[Gnanadesikan and Kettenring(1972)]{Gnanadesikan1972}
R.~Gnanadesikan and J.~R. Kettenring.
\newblock Robust estimates, residuals, and outlier detection with multiresponse
  data.
\newblock \emph{Biometrics}, pages 81--124, 1972.

\bibitem[Haldane(1948)]{Haldane1948}
J.~Haldane.
\newblock Note on the median of a multivariate distribution.
\newblock \emph{Biometrika}, 35\penalty0 (3-4):\penalty0 414--417, 1948.

\bibitem[Hampel(1974)]{Hampel1974}
F.~R. Hampel.
\newblock The influence curve and its role in robust estimation.
\newblock \emph{Journal of the American Statistical Association}, 69:\penalty0
  383--393, 1974.

\bibitem[Hampel et~al.(1986)Hampel, Ronchetti, Rousseeuw, and
  Stahel]{Hampel1986}
F.~R. Hampel, E.~M. Ronchetti, P.~J. Rousseeuw, and W.~A. Stahel.
\newblock \emph{{Robust statistics. The approach based on influence
  functions.}}
\newblock {Wiley Series in Probability and Mathematical Statistics. New York
  etc.: Wiley}, 1986.

\bibitem[Kemperman(1987)]{Kemperman1987}
J.~H.~B. Kemperman.
\newblock {The median of a finite measure on a {B}anach space}.
\newblock In Y.~Dodge, editor, \emph{Statistical Data Analysis Based on the
  $L_1$-Norm and Related Methods}, pages 217--230. Amsterdam: North-Holland,
  1987.

\bibitem[Koltchinskii and Dudley(2000)]{Koltchinskii2000}
V.~Koltchinskii and R.~Dudley.
\newblock {On spatial quantiles.}
\newblock In {Korolyuk, V. et al.}, editor, \emph{Skorokhod's ideas in
  probability theory.}, pages 195--210. Kiev: Institute of Mathematics of NAS
  of Ukraine. Proc. Inst. Math. Natl. Acad. Sci. Ukr., Math. Appl. 32, 2000.

\bibitem[Locantore et~al.(1999)Locantore, Marron, Simpson, Tripoli, Zhang, and
  Cohen]{Locantore1999}
N.~Locantore, J.~Marron, D.~Simpson, N.~Tripoli, J.~Zhang, and K.~Cohen.
\newblock Robust principal component analysis for functional data.
\newblock \emph{Test}, 8\penalty0 (1):\penalty0 1--73, 1999.

\bibitem[Lopuha{\"a}(1989)]{Lupuhaa1989}
H.~P. Lopuha{\"a}.
\newblock On the relation between {S}-estimators and {M}-estimators of
  multivariate location and covariance.
\newblock \emph{The Annals of Statistics}, pages 1662--1683, 1989.

\bibitem[Ma and Genton(2001)]{MaGenton2001}
Y.~Ma and M.~G. Genton.
\newblock Highly robust estimation of dispersion matrices.
\newblock \emph{Journal of Multivariate Analysis}, 78\penalty0 (1):\penalty0
  11--36, 2001.

\bibitem[Magnus and Neudecker(1999)]{Magnus1999}
J.~R. Magnus and H.~Neudecker.
\newblock \emph{{Matrix differential calculus with applications in statistics
  and econometrics.}}
\newblock {Wiley Series in Probability and Statistics. Chichester: Wiley}, 2nd
  edition, 1999.

\bibitem[Marden(1999)]{Marden1999}
J.~I. Marden.
\newblock Some robust estimates of principal components.
\newblock \emph{Statistics \& Probability Letters}, 43\penalty0 (4):\penalty0
  349--359, 1999.

\bibitem[Maronna and Zamar(2002)]{Maronnazamar2002}
R.~A. Maronna and R.~H. Zamar.
\newblock Robust estimates of location and dispersion for high-dimensional
  datasets.
\newblock \emph{Technometrics}, 44\penalty0 (4), 2002.

\bibitem[Maronna et~al.(2006)Maronna, Martin, and Yohai]{Maronna2006}
R.~A. Maronna, R.~D. Martin, and V.~J. Yohai.
\newblock \emph{Robust statistics}.
\newblock Chichester: Wiley, 2006.

\bibitem[Milasevic et~al.(1987)Milasevic, Ducharme, et~al.]{Milasevic1987}
P.~Milasevic, G.~Ducharme, et~al.
\newblock Uniqueness of the spatial median.
\newblock \emph{The Annals of Statistics}, 15\penalty0 (3):\penalty0
  1332--1333, 1987.

\bibitem[M{\"o}tt{\"o}nen et~al.(1999)M{\"o}tt{\"o}nen, Koivunen, and
  Oja]{Mottonen1999}
J.~M{\"o}tt{\"o}nen, V.~Koivunen, and H.~Oja.
\newblock Robust autocovariance estimation based on sign and rank correlation
  coefficients.
\newblock In \emph{Higher-Order Statistics, 1999. Proceedings of the IEEE
  Signal Processing Workshop on}, pages 187--190. IEEE, 1999.

\bibitem[Paindaveine(2008)]{Paindaveine2008}
D.~Paindaveine.
\newblock A canonical definition of shape.
\newblock \emph{Statistics \& Probability Letters}, 78\penalty0 (14):\penalty0
  2240--2247, 2008.

\bibitem[Pearson(1907)]{pearson:1907}
K.~Pearson.
\newblock \emph{Mathematical contributions to the theory of evolution. XVI. On
  further methods of determining correlation}.
\newblock Drapers' company research memoirs: Biometric series. London: Dulau \&
  Co., 1907.

\bibitem[Rousseeuw(1985)]{Rousseeuw1985}
P.~J. Rousseeuw.
\newblock {Multivariate estimation with high breakdown point.}
\newblock In W.~Grossmann, G.~C. Pflug, I.~Vincze, and W.~Wertz, editors,
  \emph{{Mathematical statistics and applications, Proc.\ 4th Pannonian Symp.\
  Math.\ Stat., Bad Tatzmannsdorf, Austria, September 4-10, 1983, Vol. B}},
  pages 283--297. Dordrecht etc.: D. Reidel, 1985.

\bibitem[Rousseeuw and Croux(1993)]{RousseeuwCroux1993}
P.~J. Rousseeuw and C.~Croux.
\newblock Alternatives to the median absolute deviation.
\newblock \emph{Journal of the American Statistical Association}, 88\penalty0
  (424):\penalty0 1273--1283, 1993.

\bibitem[Sirki{\"a} et~al.(2009)Sirki{\"a}, Taskinen, Oja, and
  Tyler]{Sirkia2009}
S.~Sirki{\"a}, S.~Taskinen, H.~Oja, and D.~E. Tyler.
\newblock Tests and estimates of shape based on spatial signs and ranks.
\newblock \emph{Journal of Nonparametric Statistics}, 21\penalty0 (2):\penalty0
  155--176, 2009.

\bibitem[Stahel(1981)]{Stahel1981}
W.~Stahel.
\newblock \emph{Robust estimation: Infinitesimal optimality and covariance
  matrix estimation}.
\newblock PhD thesis, ETH Z{\"u}rich, 1981.

\bibitem[Taskinen et~al.(2006)Taskinen, Croux, Kankainen, Ollila, and
  Oja]{Taskinen2006}
S.~Taskinen, C.~Croux, A.~Kankainen, E.~Ollila, and H.~Oja.
\newblock Influence functions and efficiencies of the canonical correlation and
  vector estimates based on scatter and shape matrices.
\newblock \emph{Journal of Multivariate Analysis}, 97\penalty0 (2):\penalty0
  359--384, 2006.

\bibitem[Tyler(1987{\natexlab{a}})]{Tyler1987}
D.~E. Tyler.
\newblock A distribution-free {M}-estimator of multivariate scatter.
\newblock \emph{The Annals of Statistics}, 15\penalty0 (1):\penalty0 234--251,
  1987{\natexlab{a}}.

\bibitem[Tyler(1987{\natexlab{b}})]{Tyler1987b}
D.~E. Tyler.
\newblock {Statistical analysis for the angular central Gaussian distribution
  on the sphere.}
\newblock \emph{Biometrika}, 74:\penalty0 579--589, 1987{\natexlab{b}}.

\bibitem[Tyler(2010)]{Tyler2010}
D.~E. Tyler.
\newblock A note on multivariate location and scatter statistics for sparse
  data sets.
\newblock \emph{Statistics \& Probability Letters}, 80\penalty0 (17):\penalty0
  1409--1413, 2010.

\bibitem[Visuri(2001)]{Visuri2001}
S.~Visuri.
\newblock \emph{Array and multichannel signal processing using nonparametric
  statistics}.
\newblock PhD thesis, Helsinki University of Technology, 2001.

\bibitem[Visuri et~al.(2000)Visuri, Koivunen, and Oja]{Visuri2000}
S.~Visuri, V.~Koivunen, and H.~Oja.
\newblock Sign and rank covariance matrices.
\newblock \emph{Journal of Statistical Planning and Inference}, 91\penalty0
  (2):\penalty0 557--575, 2000.

\bibitem[Visuri et~al.(2001)Visuri, Oja, and Koivunen]{Visuri2001a}
S.~Visuri, H.~Oja, and V.~Koivunen.
\newblock Subspace-based direction-of-arrival estimation using nonparametric
  statistics.
\newblock \emph{Signal Processing, IEEE Transactions on}, 49\penalty0
  (9):\penalty0 2060--2073, 2001.

\bibitem[Vogel et~al.(2008)Vogel, K{\"o}llmann, and Fried]{Vogel2008}
D.~Vogel, C.~K{\"o}llmann, and R.~Fried.
\newblock Partial correlation estimates based on signs.
\newblock In \emph{Proceedings of the 1st Workshop on Information Theoretic
  Methods in Science and Engineering. TICSP series}, volume~43, 2008.

\end{thebibliography}

}

\end{document}